\documentclass[11pt,a4paper]{article}
\usepackage{indentfirst}
\usepackage[latin1]{inputenc}
\usepackage{mathtools}
\usepackage{amsfonts}
\usepackage[english]{babel}
\usepackage{amssymb}
\usepackage{graphicx}
\usepackage{geometry}
\usepackage{amsthm}
\usepackage{amsmath,bm}
\usepackage{breqn}
\usepackage{float}
\usepackage{appendix}
\newtheorem{theorem}{Theorem}[section]

\usepackage{tikz} 
\usetikzlibrary{angles}
\usetikzlibrary{automata,positioning}
\usepackage{titlesec}
\titlelabel{\thetitle.\quad}
\usepackage{pgfplots}
\usepackage{hyperref}
\usepackage{titlesec}

\begin{document}
	\title{Modified General Relativity and dark matter}
	\author{Gary Nash
		\\University of Alberta, Edmonton, Alberta, Canada,T6G 2R3\\gnash@ualberta.net \footnote{Present address Edmonton, AB.}
	}
\date{~April 4, 2023}
\maketitle
\vspace{2mm}
\begin{abstract}
Modified General Relativity (MGR) is the natural extension of General Relativity (GR). MGR explicitly uses the smooth regular line element vector field $(\bm{X},-\bm{X}) $, which exists in all Lorentzian spacetimes, to construct a connection-independent symmetric tensor that represents the energy-momentum of the gravitational field. It solves the problem of the non-localization of gravitational energy-momentum in GR, preserves the ontology of the Einstein equation, and maintains the equivalence principle. The line element field provides MGR with the extra freedom required to describe dark energy and dark matter. An extended Schwarzschild solution for the matter-free Einstein equation of MGR is developed, from which the Tully-Fisher relation is derived, and the gravitational energy density is calculated. The mass of the invisible matter halo of galaxy NGC 3198 calculated with MGR is identical to the result obtained from GR using a dark matter profile. Although dark matter in MGR is described geometrically, it has an equivalent representation as a particle with the property of a vector boson or a pair of fermions; the geometry of spacetime and the quantum nature of matter are linked together by the unit line element covectors that belong to both the Lorentzian metric and the spin-1 Klein-Gordon wave equation. The three classic tests of GR provide a comparison of the theories in the solar system and several parts of the cosmos.  MGR provides the flexibility to describe inflation after the Big Bang and galactic anisotropies.
\end{abstract}
\vspace*{0.4cm}
\newline This is an update of an article published in International Journal of Modern Physics D. The original authenticated version is available online at: https://doi.org/10.1142/S0218271823500311. 
\section{Introduction}
Einstein developed his equation of General Relativity (GR)
\begin{equation}\label{E}
	G_{\alpha\beta}=\frac{8\pi G}{c^{4}}T_{\alpha\beta}
\end{equation}
in a four-dimensional Riemannian spacetime with the understanding that spacetime is locally Minkowskian during free-fall. He believed that gravity gravitates and postulated \cite{1,2} that the total energy-momentum tensor consisted of the energy-momentum from ponderable matter plus that of the gravitational field itself.
\par   However, he could not construct a \emph{tensor} that described the energy-momentum of the gravitational field and introduced a pseudo-tensor instead. Later in 1922, he concluded \cite{3} that although the true form of the energy-momentum tensor is undetermined, it must contain the energy density of the gravitational field in addition to the energy density of matter. \par Other pseudo-tensors and approaches to describe gravitational energy-momentum were subsequently investigated for decades. Over the last century, failure to find a tensor that represents the energy-momentum of the gravitational field either led relativists to the generally accepted conclusion that it does not exist \cite{4,5}, or to the notion that GR does not need a tensor that describes the self-interaction of the gravitational field because the nonlinearity in the Einstein tensor provides for that. However, the second perspective leads to the same conclusion as that of the first. By writing the Lorentzian metric as a background metric plus a tensor representing a small perturbation around the background, the Ricci tensor can then be expanded to second order in the perturbed metric. From that, an energy-momentum tensor that depends only on the background gravitational field and contains terms quadratic in the perturbed metric and its first two covariant derivatives is obtained \cite{6}. While this tensor is adequate for the description of gravitational waves, it is not a proper description of the energy-momentum of the gravitational field because it is not invariant when the connection coefficients vanish as in free-fall. \par Gravitational energy-momentum is invariant during free-fall; it has the same value when expressed covariantly as it does during free-fall when the connection coefficients locally vanish. However, a connection-independent local expression for gravitational energy-momentum does not exist in GR. That is a consequence of the equivalence principle whereby the connection coefficients, and therefore gravity, locally vanishes during free-fall because it is always possible to construct a coordinate system where all of the first derivatives of the metric vanish at any point; no connection coefficients means no gravitational field and no local gravitational field means no local gravitational energy-momentum \cite{7}. \par The right-hand side of the Einstein equation is generally considered to represent all forms of matter and associated fields, but it is void of a tensorial expression for gravitational energy-momentum. That leads to a second issue with the Einstein equation from the fact that $ G_{\alpha\beta} $ describes the nonlinear interactions of the metric with itself and its first two derivatives, but on the other side of (\ref{E}), there is nothing whatsoever in a matter energy-momentum tensor that geometrically describes a dynamically changing metric. Einstein described his discontent with the meaning of his equation \cite{8} in 1936 as ``a building, one wing of which is made of fine marble (left side of the equation), but the other wing of which is built of low-grade wood (right side of the equation)." Thus, Einstein's equation is left with two fundamental issues to be resolved: the non-localization of gravitational energy-momentum and the inconsistent ontology of his equation. 
\par Fundamental to the solution to these problems, as developed in Modified General Relativity \cite{9} (MGR), was to re-instate Einstein's original postulate of a total energy-momentum tensor $ T_{\alpha\beta} $, which must contain a tensor that geometrically describes a dynamically changing metric. Gravity gravitates and this postulate maintains the ontology of his equation. This, in turn, demands the local conservation of total energy-momentum with the vanishing of $ \nabla^{\alpha}T_{\alpha\beta} $ and not $ \nabla^{\alpha}\tilde{T}_{\alpha\beta} $, where $ \tilde{T}_{\alpha\beta} $ is the matter component of $ T_{\alpha\beta} $. Since $ \tilde{T}_{\alpha\beta} $ is not divergenceless, it can be set proportional to an arbitrary non-divergenceless symmetric tensor $ w_{\alpha\beta} $, which is then orthogonally decomposed by the Orthogonal Decomposition Theorem (ODT) (a revised version of which is in Appendix A for reference). In addition to the Lovelock tensors $ g_{\alpha\beta} $ and $G_{\alpha\beta}  $, this introduces a new tensor, $ \varPhi_{\alpha\beta} $, which represents the energy-momentum of the gravitational field and \emph{completes} GR. 
\par $ \varPhi_{\alpha\beta} $ is constructed from the Lie derivative of both a Lorentzian metric and a product of non-vanishing covectors that explicitly belong to that metric, which generates the change of those dynamic fields along the flow of the line element vector $\bm{X}=f\bm{u}$ collinear with $\bm{u}$:
\begin{equation}\label{Phiab}
	\begin{split}
		\varPhi_{\alpha\beta}:&=	\frac{1}{2}\pounds_{X}g_{\alpha\beta}+\pounds_{X}u_{\alpha}u_{\beta}\\
		&=\frac{1}{2}(\nabla_{\alpha}X_{\beta}+\nabla_{\beta}X_{\alpha})+u^{\lambda}(u_{\alpha}\nabla_{\beta}X_{\lambda}+u_{\beta}\nabla_{\alpha}X_{\lambda}).
	\end{split}
\end{equation} That change may be immeasurably small, but it is still non-zero because even weak gravitational fields gravitate. Moreover, a tensor constructed from the Lie derivative has the \emph{essential} property that it is independent of the connection. When the connection coefficients  vanish under free-fall, $ \varPhi_{\alpha\beta} $ is equivalently expressed in terms of partial derivatives. Local Lorentz invariance is preserved and there is no conflict with the equivalence principle. Gravitational energy-momentum is invariant under free-fall and can be localized.  
\par It is important to emphasize that the vector field $\bm{u}$ employed in MGR has not been introduced arbitrarily like the vectors in the vector-tensor theories of Refs. \cite{10} and \cite{11}\,; it is already a fundamental part of a particular Lorentzian spacetime. Therefore, a discussion of line element vector fields and Lorentzian metrics is necessary to understand 
the structure of MGR.
\par As a first attempt to understand how the regular line element vectors provide the freedom to explain dark matter, simple ansatzes were used in Ref. \cite{9} for the vector components without investigating all related equations. While it was illuminating to discover the form of the extended Schwarzschild metric, the ansatzes used were far too simple.\par Thus, this paper is organized as follows. In Section 2, the structure of a Lorentzian spacetime is discussed, and a partial review of MGR is presented. Section 3 constructs the static spherically symmetric field equations in matter-free spacetime. Power series expressions for the two independent radial static components of $ \bm{X} $, $ X_{1} $ and $ X_{2} $, are thoroughly developed, from which the extended Schwarzschild solution is obtained. Conservation of gravitational energy during free-fall is discussed. Gravity gravitates, and the gravitational energy density is calculated to be twice the Newtonian result plus a dark energy contribution. Section 4 is a discussion of dark matter. The gravitational scaling of the dark parameter $ a_{0} $ is established, and the extended Newtonian force is developed from which the Tully-Fisher relation follows. The mass of the invisible matter halo of galaxy NGC 3198 is calculated to be 8.6$ \times10^{11}M\odot $, which is identical to that obtained using GR and the Navarro-Frenk-White (NFW) dark matter profile. That value was used to determine the parameter $ \zeta $, which is fundamental to the scaling of $ a_{0} $. It is shown that dark matter can be described both geometrically and as a particle with the property of a spin-1 unit vector boson or a pair of Hermitian spin 1/2 fermions. In Section 5, the three classic tests of General Relativity are recalculated in the extended Schwarzschild metric of MGR. The total perihelion advance of Mercury contains a correction from dark matter that is approximately an order of magnitude less than the contribution from the quadrupole moment of the Sun. The gravitational lensing of galactic clusters in MGR provides significant corrections to the GR results with no dark matter profile. In particular, the strong lensing in galactic cluster SDSS J0900+2234 calculated in MGR is comparable to that determined in GR with the NFW dark matter profile. The Shapiro time delay of the Crab pulsar calculated in MGR is larger than that obtained in GR from dark matter profiles commensurate with the fact that those dark matter profiles are not valid for distances from the galactic center of the Milky Way in the interval  $[5,30]$\;kpc. Cosmological aspects of MGR are discussed in Section 6. A comparison of the $\Lambda$CDM standard model of cosmology to MGR is presented. The anisotropy observed in galactic clusters, the cosmic microwave background (CMB), and the dark matter skeleton of curved spacetime are discussed relative to MGR. 
\section{Lorentzian spacetimes and Modified General Relativity}
\subsection{Lorentzian spacetimes}
A Lorentzian curved spacetime is described on a manifold $ \mathcal{M} $ with a Lorentzian metric $(\mathcal{M},g_{\alpha\beta})$. The manifold is endowed with four important topological properties: Hausdorff, connected, non-compact paracompact (paracompact), and boundaryless. A Hausdorff manifold demands that any two distinct points are separated by disjoint open neighbourhoods, which allows for physically distinct events to exist. Connectedness implies that communication is possible between separated entities. A manifold is paracompact if every open covering admits a locally finite refinement $\{ V_{\beta}\} $; that is, every point has an open neighbourhood meeting only finitely many $ V_{\beta} $'s. On such manifolds, the existence of closed causal curves is not guaranteed, contrary to the case when the manifold is compact. A manifold with a boundary would imply that there exists a physical ``barrier" which has not been observed in nature. Thus, the manifold describing a spacetime in this paper is considered to be Hausdorff, connected, paracompact and smooth without boundary.\par These properties of the manifold lead to the structure of a Lorentzian metric. A smooth regular vector field $ \bm{X} $ exists \cite{12,13} on any paracompact Hausdorff manifold $ \mathcal{M} $, as does the smooth regular unit vector $ \bm{u} $ collinear to $ \textbf{X} $: $\bm{X}=f\bm{u}  $ where $ f\neq0 $ is the scalar magnitude of $ \bm{X}$. The line element (or direction) field $(\bm{X},-\bm{X}) $ is defined as an assignment of a pair of equal and opposite vectors at each point of $\mathcal{M}$. It is like a vector field with an undetermined sign and is a one-dimensional vector subspace of the tangent space on $ \mathcal{M} $. It is well known that timelike, null and spacelike vectors belong to the tangent space of a Lorentzian manifold. The subspace of the tangent space where the line element field exists has been mainly used in GR to develop theorems \cite{14,15} on causality. However, the line element field is fundamental to the existence of a Lorentzian metric: a paracompact manifold admits a Lorentzian metric $ g_{\alpha\beta} $ if and only if it admits a line element field  \cite{14}.  \par A Lorentzian metric with a +2 signature is constructed from a Riemannian metric $ g^{+} $, which always exists on $\mathcal{M}$, and a unit line element vector $ \bm{u} $ \cite{12,13,14,15}:
\begin{equation}\label{gab}
	g_{\alpha\beta}=g_{\alpha\beta}^{+}-2u_{\alpha}u_{\beta}
\end{equation}where the Riemannian metric $g+$ is independent of $\bm{u}$. Thus, a Lorentzian spacetime does not exist without a line element field. It is not enough to say a vector field exists on $ \mathcal{M} $; a Lorentzian metric must have a directional attribute, which is what the line element field provides.\par There are an infinite number of Riemanninan metrics and an infinite number of line element vectors on $\mathcal{M}$ so a Lorentzian metric is highly nonunique. Nevertheless, there is a definite relationship between a particular
Lorentzian metric g and the vectors $\bm{u}$ and $\bm{X}$. MGR exploits the covectors explicit in a particular Lorentzian metric to construct the symmetric connection-independent tensor $\varPhi_{\alpha\beta} $ from the ODT; the covectors  $ u_{\alpha}$ and $u_{\beta} $ in (\ref{gab}) are precisely those in $\varPhi_{\alpha\beta} $. The flow vector $ \bm{X} $ in the Lie derivatives of (\ref{Phiab}) is collinear with $\bm{u}$. Thus, the vectors $ \bm{u} $ and $ \bm{X} $ employed in MGR are directly related to a given Lorentzian metric and are not introduced arbitrarily. Solutions to the Einstein equation now depend on $ \bm{u} $ and $ \bm{X} $, which provide the extra freedom required to describe dark energy and the invisible mass attributed to dark matter. GR is restricted to the tangent space of the Lorentzian manifold $\mathcal{M}$ because it does not explicitly depend on the line element field vectors.
\par Furthermore, on any paracompact Hausdorff manifold there exists a Lorentzian structure that does not generate closed timelike curves \cite{12}. The spacetime is assumed to admit a Cauchy surface and is therefore globally hyperbolic. This forbids the presence of closed causal curves. 
\subsection{Modified General Relativity}
The orthogonal decomposition of an arbitrary non-divergencless symmetric tensor $ w_{\alpha\beta} $ with the ODT (Appendix A) yields $w_{\alpha\beta} =v_{\alpha\beta}+\varPhi_{\alpha\beta} $. $ v_{\alpha\beta} $ is a linear sum of divergenceless symmetric tensors, which sum consists of the Lovelock tensors and other possible non-Lovelock tensors. In a four-dimensional spacetime, Lovelock's theorem \cite{Love} states that the only tensors which are symmetric, divergenceless, and a concomitant of the metric together with its first two derivatives, are the metric and the Einstein tensor.  \par As explained in the introduction, setting $ w_{\alpha\beta}=\frac{8\pi G}{c^{4}}\tilde{T}_{\alpha\beta}=v_{\alpha\beta}+\varPhi_{\alpha\beta} $ provides for the orthogonal decomposition of the matter energy-momentum tensor $\tilde{T}_{\alpha\beta} $. Since $T_{\alpha\beta} $ is uniquely determined up to a divergenceless tensor, the sum of possible divergencless non-Lovelock tensors in the Einstein equation can be dismissed. Hence, $v_{\alpha\beta}=\Lambda g_{\alpha\beta}+G_{\alpha\beta}$ where $\Lambda$ is the cosmological constant. Thus, \emph{the ODT and Lovelock's theorem generates the complete Einstein equation in one line}: $\Lambda g_{\alpha\beta}+G_{\alpha\beta}+\varPhi_{\alpha\beta}=\frac{8\pi G}{c^{4}}\tilde{T}_{\alpha\beta} $. It follows that MGR does not introduce any third or higher order derivative terms into the Einstein equation.
\par For clarity and completeness, some structural aspects of MGR \cite{9} are briefly discussed. The action functional $ S=S^{F}+S^{EH}+S^{G} $ consists of the action for all matter fields $S^{F}$, the Einstein-Hilbert action of general relativity $S^{EH}$ and the action for the energy-momentum of the gravitational field, $S^{G}=-a\int\Phi\sqrt{-g}d^{4}x$:
\begin{equation}\label{Se}
	\begin{split}
		S=\int L^{F}( A^{\beta},\nabla^{\alpha} A^{\beta},...,g^{\alpha\beta})\sqrt{-g}d^{4}x
		+a \int (R-2\Lambda)\sqrt{-g}d^{4}x\\-a\int \varPhi_{\alpha\beta} g^{\alpha\beta}\sqrt{-g} d^{4}x
	\end{split}
\end{equation}where $L^{F}  $ is the Lagrangian of the matter fields, and $ \Phi $ is the trace of $ \varPhi_{\alpha\beta} $ with respect to the inverse metric. The details of the variations of the action functional $S$ with respect to the variables $g^{\alpha\beta}$, $ u^{\nu}$ and $f $ are given in Appendix B.  With $a=\frac{c^{3}}{16\pi G}$, the variation of $S$ with respect to $g^{\alpha\beta}$ generates the modified Einstein equation
\begin{equation}\label{MEQ}
	G_{\alpha\beta}+{\Lambda} g_{\alpha\beta}+\varPhi_{\alpha\beta}=\frac{8\pi G}{c^{4}}\tilde{T}_{\alpha\beta}
\end{equation} and 
\begin{equation}\label{nuab}
	\nabla_{\alpha}(u^{\alpha}u^{\beta})=0.
\end{equation}
Using (\ref{nuab}), the integral of $ \Phi $ over all spacetime vanishes:
\begin{equation}\label{intPhi}
	\int \Phi \sqrt{-g}d^{4}x=0.
\end{equation}As discussed in Ref. \cite{9} for a closed Universe in the Friedmann-Lema\^{i}tre-Robertson-Walker (FLRW) metric, $ \Phi>-2\varPhi_{00} $ is gravitationally repulsive and describes dark energy. The opposite inequality describes a gravitationally attractive scenario. However, these results are independent of the parameter $ \kappa $ which describes an open, closed or flat Universe in the FLRW metric. Thus, (\ref{intPhi}) keeps a homogeneous and isotropic Universe in balance. Moreover, (\ref{intPhi}) applies to any Universe described by a Lorentzian metric; it cannot rip apart or contract to oblivion. In general, dark energy in a Lorentzian spacetime is described by the condition 
\begin{equation}\label{de}
	\Phi>-2\varPhi_{00}. 
\end{equation} The energy-momentum tensor of the gravitational field exhibits important cosmological properties. \par  Although the integral $ S^{G} $ vanishes, it admits the variation:\newline $ \delta S^{G}=a\int\varPhi_{\alpha\beta}\delta g^{\alpha\beta}\sqrt{-g}d^{4}x $.
It then follows from (\ref{intPhi}), that the action
\begin{equation}\label{key}
	S^{EHG}=\frac{c^{3}}{16\pi G}\int (R-\Phi)\sqrt{-g}d^{4}x
\end{equation} generates the modified Einstein equation with no cosmological constant. If $ \Phi$ is set to the constant
2$\Lambda $, the Einstein equation with a cosmological constant is obtained accordingly, which contradicts (\ref{intPhi}). Thus, $ \Phi $ dynamically replaces the cosmological constant in the action functional. Einstein's equation in MGR 
\begin{equation}\label{ME}
	G_{\alpha\beta}=\frac{8\pi G}{c^{4}}T_{\alpha\beta}
\end{equation} with the total energy-momentum tensor given by
\begin{equation}\label{Tab}
	T_{\alpha\beta}=\tilde{T}_{\alpha\beta}-\frac{c^{4}}{8\pi G}\varPhi_{\alpha\beta}
\end{equation} is identical in form to that of GR, but is now complete by having a tensor that represents the energy-momentum of the gravitational field. The local conservation law, $ \nabla^{\alpha}T_{\alpha\beta}=0 $, follows from the diffeomorphic invariance of MGR. In the vacuum, it follows from (\ref{Tab}) that $ \nabla^{\alpha}\varPhi_{\alpha\beta}=0 $ and the contracted Bianchi identity is satisfied from (\ref{EF}) as required.\par The action functionals $S^{F}$ and $S^{EH}$ do not depend on $X^{\mu}$. Variation of the action functional $S^{G}$ in (\ref{Se}) with respect to the dynamic variable $ u^{\nu} $ gives
\begin{equation}\label{unu}
	u_{\nu}=\frac{\partial_{\nu}f}{\Phi}, 
\end{equation} where $ f\neq0 $ is the magnitude of $ X^{\beta} $ dual to $ X_{\nu}=fu_{\nu} $ by $ g^{\beta\nu} $. These are the covectors that describe dark matter as discussed in subsection 4.7. \par Since the magnitude of $ X^{\beta} $ is independent of $ u^{\beta} $, variation of $S^{G}$ in (\ref{Se}) with respect to $ f $ generates 
\begin{equation}\label{nabu}
	\nabla_{\alpha}u^{\alpha}=0, 
\end{equation}
which is also obtained from (\ref{nuab}) in an affine parameterization. It follows from (\ref{nabu}) and (\ref{unu}) that $ f $ satifies the inhomogeneous wave equation 
\begin{equation}\label{wf}
	\square f=u^{\alpha}\partial_{\alpha}\Phi=\pounds_{u}\Phi
\end{equation}where $ \square f:=\nabla_{\alpha}\nabla^{\alpha}f $. Since the unit vectors $ u^{\alpha} $ in a Lorentzian spacetime satisfy 
\begin{equation}\label{uu1}
	u^{\alpha}u_{\alpha}=-1
\end{equation}and (\ref{nabu}), there are two independent components of $ u^{\alpha}=\frac{\partial^{\alpha}f}{\Phi} $ generated from (\ref{unu}). It follows from (\ref{wf}) and $ X_{\alpha}X^{\alpha}=-f^{2} $ that $ \bm{X} $ has two independent components, which is consistent with $ u^{\alpha} $ being collinear with $ X^{\alpha} $.
\par  $\varPhi_{\alpha\beta}  $ vanishes if and only if $ X^{\mu} $ is a Killing vector: if $ X^{\mu} $ is a Killing vector, $\nabla_{\alpha}X_{\beta}+\nabla_{\alpha}X_{\alpha}=0$ from which $\varPhi_{\alpha\beta}=-\frac{u_{\alpha }X^{\lambda}\nabla_{\lambda}X_{\beta}}{f}-\frac{u_{\beta }X^{\lambda}\nabla_{\lambda}X_{\alpha}}{f}=0$ in an affine parameterization; conversely, if the tensor $\varPhi_{\alpha\beta}=0$, it vanishes in all coordinate systems including that defined by $u^{\lambda}=(1,0,0,0)$ and $u_{\lambda}=(-1,0,0,0)$. Then $0=\varPhi_{\alpha\beta}=\frac{1}{2}(\nabla_{\alpha}X_{\beta}+\nabla_{\beta}X_{\alpha})-(\nabla_{\beta}X_{\alpha}+\nabla_{\alpha
}X_{\beta})$ and $X^{\beta}$ is a Killing vector. However, in general, there are no Killing vector fields unless a particular symmetry is involved.

\section{The extended Schwarzschild solution and the energy density of the gravitational field}
In a region of spacetime where there is no matter of any kind whatsoever, the matter
Lagrangian in (\ref{Se}) vanishes. $ \tilde{T}_{\alpha\beta}=0 $ and the field equations in the vacuum are:
\begin{equation}\label{EF}
	G_{\alpha\beta}+\varPhi_{\alpha\beta}=0.
\end{equation}  Spherically symmetric solutions to these nonlinear equations are now investigated with a metric of the form
\begin{equation}\label{g}
	ds^{2}=-e^\nu c^{2}dt^{2}+e^{\lambda}dr^{2}+r^{2}(d\theta^{2}+sin^{2}\theta d\varphi^{2})
\end{equation} where $ \nu $ and $ \lambda $ are functions of $r$ and $t$. 
Since the metric is spherically symmetric, there exist three Killing vectors associated with the spherical symmetry of a three-dimensional spatial rotation. $ \varPhi_{\alpha\beta} $ vanishes when $ X^{\beta} $ is a Killing vector and (\ref{EF}) reduces to $ R_{\alpha\beta}=0 $. It follows that $ \nu=-\lambda $ holds in MGR as a direct result of the spherical symmetry invoked to solve (\ref{EF}).
\par  Static solutions to (\ref{EF}) are now sought, which requires $ \partial_{0}X_{\alpha}=0 $ and from the metric $ \partial_{0}\lambda=0,\enspace \partial_{0}\nu=0$. It follows that $ u_{\alpha}\partial_{0}f+f\partial_{0}u_{\alpha}=0 $ and (\ref{uu1}) holds. 
\par With $X_{3}=0$ and $\mu:=(1+2u_{0}u^{0})$, the components of $ \varPhi_{\alpha\beta} $ to be considered are then: 
\begin{equation}\label{Phi00}
	\varPhi_{00}=\frac{\mu}{2}e^{-2\lambda}\lambda^{\prime} X_{1}
\end{equation}
\begin{equation}\label{key}
	\varPhi_{11}=(1+2u_{1}u^{1} )({X_{1}}^{\prime}-\frac{1}{2}\lambda^{\prime}X_{1}),
\end{equation} 
\begin{equation}\label{key}
	\varPhi_{22}=(1+2u_{2}u^{2} )(\partial_{2}X_{2}+re^{-\lambda} X_{1}),  
\end{equation} 
\begin{equation}\label{key}
	\varPhi_{33}=r \sin^{2}\theta e^{-\lambda}X_{1}+\sin\theta \cos\theta X_{2},
\end{equation} the Ricci scalar, which from (\ref{EF}) equals $ \Phi $, is
\begin{equation}\label{R}
	\begin{split}
		R=e^{-\lambda}(\lambda^{\prime\prime}-{\lambda^{\prime}}^{2}+\frac{4}{r}\lambda^{\prime}-\frac{2}{r^{2}})+\frac{2}{r^{2}},
	\end{split}
\end{equation}
and the corresponding components of the Einstein tensor are:
\begin{equation}\label{key}
	\begin{split}
		G_{00}=\frac{1}{r^{2}}e^{-2\lambda}(r\lambda^{\prime}-1)+\frac{e^{-\lambda}}{r^{2}},
	\end{split}
\end{equation}
\begin{equation}\label{key}
	\begin{split}
		G_{11}=\frac{1}{r^{2}}(1-r\lambda^{\prime})-\frac{e^{\lambda}}{r^{2}},
	\end{split}
\end{equation}
\begin{equation}\label{key}
	\begin{split}
		G_{22}=\frac{r^{2}e^{-\lambda}}{2}(-\lambda^{\prime\prime}+{\lambda^{\prime}}^{2}-\frac{2\lambda^{\prime}}{r}),
	\end{split}
\end{equation}
\begin{equation}\label{key}
	\begin{split}
		G_{33}={\sin\theta}^2[\frac{r^{2}e^{-\lambda}}{2}(-\lambda^{\prime \prime}+{\lambda^{\prime}}^{2}-\frac{2\lambda^{\prime}}{r})]
	\end{split}
\end{equation} where the prime denotes $ \partial_{1} $.  \par Since $ e^{2\lambda}(\varPhi_{00}+G_{00})+\varPhi_{11}+G_{11}=0 $ from (\ref{EF}),
\begin{equation}\label{phiG0011}
	\begin{split}
		\mu\frac{\lambda^{\prime}}{2}X_{1}+(X_{1}^{\prime}-\frac{\lambda^{\prime}}{2}X_{1})(1+2u_{1}u^{1})=0
	\end{split}
\end{equation} where $\mu:=(1+2u_{0}u^{0})$.
\par  $G_{22}+\varPhi_{22}=0 $ gives \begin{equation}\label{phiG22}
	\begin{split}
		-\lambda^{\prime\prime}+\lambda^{\prime2}-\frac{2}{r}\lambda^{\prime}+\frac{2e^{\lambda}}{r^{2}}(\partial_{2}X_{2}+re^{-\lambda}X_{1})(1+2u_{2}u^{2})=0
	\end{split}
\end{equation}

and $G_{33}+\varPhi_{33}=0 $ in the interval $ 0<\theta<\pi $ yields
\begin{equation}\label{phiG33}
	\begin{split}
		-\lambda^{\prime\prime}+\lambda^{\prime 2}-\frac{2}{r}\lambda^{\prime}+\frac{2e^{\lambda}}{r^{2}}(re^{-\lambda}X_{1}+X_{2}\cot\theta)=0.
	\end{split}
\end{equation} Subtracting (\ref{phiG22}) from (\ref{phiG33}) requires
\begin{equation}\label{C2}
	X_{2}\cot\theta -\partial_{2}X_{2}-2u_{2}u^{2}(\partial_{2}X_{2}+re^{-\lambda}X_{1})=0.
\end{equation} 
\par Using (\ref{uu1}), an equation involving the two independent components $ X_{1} $ and $ X_{2} $ can be obtained from (\ref{phiG0011}) and (\ref{C2}):
\begin{equation}\label{C}
	\begin{split}
		X_2\cot\theta-\partial_{2}X_{2}+\mu(1+\frac{\lambda^{\prime}X_{1}}{\lambda^{\prime}X_{1}-2X_{1}^{\prime}})(\partial_{2}X_{2}+re^{-\lambda}X_{1})=0.
	\end{split}
\end{equation}This is the equation from which power series expressions for $ X_{1} $ and $ X_{2} $ are sought, which can then be used in (\ref{phiG33}) to generate $ \lambda $. \par An expression for $ X_{1}$ of the form $ X_{1}=e^{\lambda}P $ is pursued where $ P $ is a polynomial in $ r $. The power series for $X_{1} $ is then assumed to be 
\begin{equation}\label{X11}
	X_{1}=e^{\lambda}(a_{0}+\frac{a_{1}}{r}+\frac{a_{2}}{r^{2}})
\end{equation} where $ a_{0} $, $ a_{1} $ and $ a_{2} $ are real arbitrary parameters. The power series for $ P $ terminates with the $\frac{a_{2}}{r^{2}}$ term, which ensures the gravitational energy density has the Newtonian $ \frac{1}{r^{4}} $ behaviour.
\par  Equation (\ref{C}), with $ X_{1} $ given by (\ref{X11}), becomes 
\begin{equation}\label{N}
	\partial_{2}X_{2}-m(r)X_{2}\cot\theta=n(r)
\end{equation}where $ n(r):=\frac{2\mu rPP^{\prime}}{\lambda^{\prime}P+2P^{\prime}(1-\mu)}$ and $m(r)=\frac{\lambda^{\prime}P+2P^{\prime}}{\lambda^{\prime}P+2P^{\prime}(1-\mu)}$, which has the solution 
\begin{equation}\label{key}
	\begin{split}
		X_{2}=\frac{n}{1-m}\sin\theta\;_{2}F_{1}(\frac{1}{2},\frac{1-m}{2};\frac{3-m}{2};\sin^{2}\theta)+c_{10}\sin^{m}(\theta).
	\end{split}
\end{equation}$_{2}F_{1}(\alpha,\beta;\gamma;z)  $ is the Gaussian hypergeometric function with $ \alpha=\frac{1}{2} $, $ \beta=-\frac{1-m}{2} $, $ \gamma=\frac{3-m}{2} $ and $ z=\sin^{2}\theta $. Setting the arbitrary constant $ c_{10} $ to zero and using the Euler transformation $_{2}F_{1}(\alpha,\beta;\gamma;z)=(1-z)^{-\alpha}\;  _{2}F_{1}(\alpha,\gamma-\beta;\gamma;\frac{z}{z-1})   $ for $0<\theta<\frac{\pi}{4}  $ yields
\begin{equation}\label{Fabc}
	\begin{split}
		X_{2}&=\frac{n}{1-m}\tan\theta\;_{2}F_{1}(\frac{1}{2},1;\frac{3-m}{2};-\tan^{2}\theta)\\
		&=-rP\tan\theta(1+\Sigma_{n=1}^{\infty}\frac{(\alpha)_{n}(\beta)_{n}(-\tan^{2}\theta)^{n}}{(\gamma)_{n}n!})\\
		&=-rP\tan\theta(1+\frac{1}{3-m}(-\tan^{2}\theta) +O(\tan^{4}\theta) )
	\end{split}
\end{equation}where $ \alpha=\frac{1}{2} $, $ \beta=1 $, $ \gamma=\frac{3-m}{2} $ and $ (\alpha)_{n}=\alpha(\alpha+1)...(\alpha+n-1)\enspace n>0 $ is the Pochhammer symbol. Using the first term of (\ref{Fabc}) as a solution for $ X_{2} $ in (\ref{phiG33}) gives the Schwarzschild solution. 
Expanding $\frac{1}{3-m}:=q(r)$ as a series of inverse powers of $r$ and taking  $\tan\theta<<1$ gives $X_{2}\simeq\tan\theta(b_{2}r+b_{0}+\frac{b_{1}}{r}+O(r^{-2}))$ where the $b_{j}$ are products or sums of products of $a_{j}$ and $q_{j}$ and are therefore independent of the $a_{j}$ in $P$. The second and higher terms of (\ref{Fabc}) consist of powers of $\tan^{2}\theta$ and can be neglected for small $ \tan\theta $. The radial dependence of the expansion of $ X_{2} $ is then then assumed to be of the form
$ b_{2}r+b_{0}+\frac{b_{1}}{r} $ where $b_{2}$, $b_{0} $ and $ b_1 $ are real arbitrary parameters. 
\par Assuming $ b_{2}\neq0$ and using (\ref{uu1})  with $X_{\alpha}=fu_{\alpha} $ where $ f\neq0 $ is the magnitude of $ X_{\alpha}$, it follows that $\frac{b_{2}^{2}\tan^{2}\theta}{f^{2}}\approx-1   $ for extremely large but finite $ r $, which is impossible. Thus, $ b_{2} $ must vanish and $ X_{2} $ does not blow-up as $ r $ becomes extremely large. The physically relevant approximate solution for $ X_{2} $ is then
\begin{equation}\label{X2}
	X_{2}=(b_{0}+\frac{b_{1}}{r})\tan\theta,\;0<\tan\theta<<1.
\end{equation} \par Equation $(\ref{phiG33})$ can now be written as 
\begin{equation}\label{phiG33a}
	\begin{split}
		-\lambda^{\prime\prime}+\lambda^{\prime 2}-\frac{2}{r}\lambda^{\prime}+\frac{2e^{\lambda}}{r^{2}}(a_{0}r+a_{1}+\frac{a_{2}}{r}+b_{0}+\frac{b_{1}}{r})=0.
	\end{split}
\end{equation}
By demanding (in hindsight) 
\begin{equation}\label{b1}
	b_{1}=-a_{2} 
\end{equation}the undesirable term with $ \frac{\ln r}{r} $ in the solution for $ \lambda $ can be avoided. Equation (\ref{phiG33a}) then simplifies to
\begin{equation}\label{main}
	\begin{split}
		-\lambda^{\prime\prime}+\lambda^{\prime 2}-\frac{2}{r}\lambda^{\prime}+\frac{2e^{\lambda}}{r^{2}}(a_{0}r-b)=0
	\end{split}
\end{equation}where
\begin{equation}\label{b}
	-b=a_{1}+b_{0},
\end{equation} which has the exact solution
\begin{equation}\label{wow}
	\begin{split}
		\lambda=-\ln(-a_{0}r+2b\ln r+\frac{c_{1}}{r}+c_{2}), \enspace 0<r<\infty
	\end{split}
\end{equation}where $c_{1}$ and $ c_{2}  $ are arbitrary parameters. Equation (\ref{wow}) represents the extended Schwarzschild solution. It is important to note that since the line element vectors are non-vanishing, $ r>0 $. Moreover, the extended Schwarzschild solution demands $ r $ to be finite; $ r $ can be as large as necessary to describe any physically reasonable Universe or part thereof, but it cannot extend to infinity.
\subsection{The energy density of the gravitational field} 
Einstein \cite{2} produced a non-tensorial object to represent the energy-momentum of the gravitational field. Subsequently, there have been many pseudo-tensors and other approaches introduced over the decades to describe gravitational energy-momentum in GR. However, pseudo-tensors generally suffer from not having well defined values in all reference frames. This led to the modern concept of quasi-local energy-momentum \cite{Chen} associated with a closed 2-surface. While quasi-locality is a far superior construct as compared to pseudo-tensors to describe gravitational energy-momentum, it cannot represent local gravitational energy-momentum because the equivalence principle forbids it. That issue is resolved with MGR because it contains a connection-independent tensor that locally represents the energy-momentum of the gravitational field; that tensor, which is missing in GR, replaces  pseudo-tensors and quasi-local energy-momentum.
\par  It is well known \cite{MM} that the local properties of the FLRW cosmology are identical to Newtonian gravitation. Equation (\ref{F2}) can be expressed as 
\begin{equation}\label{key}
	\frac{1}{2}\dot{a}^{2} -\frac{\frac{4\pi}{3}a^{3}\rho G}{a}+\frac{c^{2}a^{2}\varPhi_{00}}{6}=-\frac{\kappa c^{2}}{2},
\end{equation}
which represents: the kinetic energy of a particle with a unit mass moving with the expansion at a distance $ a $ relative to the Big Bang at $ a=0 $, plus the gravitational potential energy relative to a sphere of gravitating matter, plus the energy of the gravitational field ($ \varPhi_{00} $ has the dimensions of $ L^{-2} $), equals a constant where the constant $ \kappa $ is related to the curvature of spacetime in the FLRW metric. During free-fall, the curvature of spacetime does not change. In GR, $ \varPhi_{00} $ does not exist and the energy of a comoving particle in free-fall is conserved. Maintaining that result in MGR demands $ \varPhi_{00} $ to be invariant under free-fall; it must be connection-independent so that the energy of the gravitational field has the same value when expressed covariantly as it does during free-fall when the connection coefficients locally vanish. As defined in terms of the Lie derivative in (\ref{Phiab}), $ \varPhi_{\alpha\beta} $ is connection-independent so gravitational energy in the FLRW metric of MGR is invariant during free-fall. Moreover, that result is not limited to the FLRW metric. $ \varPhi_{\alpha\beta} $ pertains to any Lorentzian metric, so gravitational energy is invariant during free-fall in MGR for all Lorentzian metrics.
\par The energy density of the gravitational field in MGR follows from (\ref{Tab}):
\begin{equation}\label{W}
	W=-\frac{c^{4}}{8\pi G}\varPhi_{00}.
\end{equation} From $
\varPhi_{00}=\frac{\mu}{2}X_{1}e^{-2\lambda}\lambda^{\prime}$, (\ref{X11}) and (\ref{wow}), 
\begin{equation}\label{key}
	\varPhi_{00}=\frac{\mu}{2}(a_{0}-\frac{2b}{r}+\frac{c_{1}}{r^{2}})(a_{0}+\frac{a_{1}}{r}+\frac{a_{2}}{r^{2}}),
\end{equation}which must yield the $ \frac{1}{r^{4}} $ behaviour of the Newtonian gravitational energy density and include a term involving dark matter. By setting the coefficients of the $ \frac{1}{r}$ and $\frac{1}{r^{3}} $ terms to zero, we obtain 
\begin{equation}\label{a1}
 2b=a_{1} 	
\end{equation}
 and $ 2ba_{2}=c_{1}a_{1} $, respectively, from which
\begin{equation}\label{a2}
	a_{2}=c_{1}.
\end{equation} A contribution from dark matter must come from the $\frac{1}{r^{2}}$ term, which vanishes if $a_{1}^{2}=2a_{0}\mid c_{1}\mid$. Thus, to prevent that, we set 
\begin{equation}
	2a_{0}\mid c_{1}\mid=\upsilon a_{1}^{2} 
\end{equation}where $\upsilon>1$ is a parameter that relates dark matter $M_{DM}$ to ordinary matter $M$ by the ratio $M_{DM}=\upsilon M $. From (\ref{a1}) and (\ref{a2}),
\begin{equation}\label{b}
	b=\pm\sqrt{\frac{a_{0}\mid c_{1}\mid\upsilon}{2}}
\end{equation}where $ c_{1} $ is chosen to be the parameter
\begin{equation}\label{c1}
c_{1}=-\frac{2GM}{c^{2}}
\end{equation}from the Schwarzschild solution. Hence,
\begin{equation}\label{gen}
	\varPhi_{00}=\mu(\frac{1}{2}a_{0}^{2}+\frac{4b^{2}(\upsilon-1)}{r^{2}}+\frac{c_{1}^{2}}{2r^{4}})
\end{equation} \par It is often the case that the amount of dark matter relative to ordinary matter is not known, and $\upsilon$ cannot be determined. Nevertheless, the line element vectors exist at all points in spacetime, and dark matter exists everywhere. When $\upsilon$ is not known, the mass of a system still contains an observationally unknown (small) amount of dark matter, and $M$ is considered to be the gravitating mass.
\par From (\ref{W}), it follows that the energy density of the static gravitational field is
\begin{equation}\label{W0}
	W=-\frac{\mu c^{4}}{16\pi G}(a_0^{2}+\frac{8b^{2}(\upsilon-1)}{r^{2}}+\frac{4G^{2}M^{2}}{r^{4}})\;\;\upsilon>1.
\end{equation}
If this static spherically symmetric system satisfies $u_{0}u^{0}=u_{i}u^{i}\;\;i=1,2,3 $\:\:no sum on i, then $u_{0}u^{0}=-\frac{1}{4}$ and $\mu=\frac{1}{2}$. The last term equals the Newtonian gravitational energy density and that calculated in GR from the weak field approximation: $-\frac{GM^{2}}{8\pi r^{4}}\frac{J}{m^{3}} $, which has the value -$2.25\times10^{4}\frac{J}{m^{3}}$ on the Earth relative to the Sun. Using the value of $a_{0}=5.737\times10^{-41}$ for the calculation of the dark matter of the Sun from (\ref{a0zeta}), the first term has the value $-7.96\times10^{-39}\frac{J}{m^{3}}$, which is generally too small to be measurable. The second term is the dark matter term, which has the value $-7.57\times10^{-16}\frac{J}{m^{3}}  $ with $\upsilon=5$.
\par In an orthonormal basis $(\textbf{e}_{\alpha})$ for $g_{\alpha\beta}$ such that $\textbf{e}_{0}=\textbf{u}$,  $u^{\alpha}=\delta^{\alpha}_{0} $ and $u_{0}=-1$, which demands $\mu=-1$. Gravity gravitates and we see that the gravitational energy density is positive and twice the Newtonian gravitational energy density, and twice that calculated in GR from the weak field approximation. This is not surprising because in GR, it is calculated from the Newtonian potential $\phi  $ in terms of $ h_{\alpha\beta} $, a minute change in the metric relative to the flat spacetime Minkowski metric in the linearized field equations \cite{LL}. The perturbation expansion of the field equations in terms of $ h_{\alpha\beta} $ contains an infinite number of minute terms involving $ h^{2},h^{3},...h\partial h,...h\partial\partial h... $ to all possible powers of $ h$ and its first two derivatives. These truncated self-interactions of the perturbed field apparently contribute the same amount of gravitational energy to $W$. However, there is no way of knowing that in GR because it has no tensor that explicitly represents the energy-momentum of the gravitational field. That $W$ is positive in a comoving frame is necessary to describe the local energy density of a gravitational wave: $W=(\frac{GM^{2}}{4\pi r^{4}}+\frac{b^{}c^{4}(\upsilon-1)}{2\pi Gr^{2}}+\frac{c^{4}a_0^{2}}{16\pi G })\frac{J}{m^{3}}$. 
\section{Dark matter}
The existence of dark or invisible matter is based on the assumption that GR is a complete theory. However, it lacks a tensor that describes the energy-momentum of the gravitational field and is not complete. GR cannot successfully describe the flat rotation curves of most galaxies without introducing additional mass that appears to not interact with the electromagnetic field. Whereas, MGR contains the gravitational energy-momentum tensor $ \varPhi_{\alpha\beta} $ that is missing in GR. That tensor is constructed from the line element field vectors. The ability of the line element field to describe dark matter is now investigated.
\subsection{The dark parameters $ a_{0} $ and $ b $ of the line element field}
Since every point in a Lorentzian spacetime admits a pair of equal and opposite line element vectors, both $ a_{0} $ and $ b $ have two signs, as do all other parameters in (\ref{X11}) and (\ref{X2}). The line element field consists of regular vectors so $ a_{0} $ and $ b $ never vanish. The parameter $ c_{2} $ in the  extended Schwarzschild solution (\ref{wow}) is chosen to be unity in accordance with the Schwarzschild solution of GR. There is one independent parameter, $ a_{0} $, once $ c_{1} $ is fixed; all other parameters can be determined in terms of $ a_{0} $ and $ c_{1}=-\frac{2GM}{c^{2}} $ where $ M $ represents the ordinary mass of the system if $\upsilon$ is known; otherwise $M$ is the gravitating mass. Since $ b $ is determined from $a_{0}  $ and $ c_{1} $, the free parameter $ a_{0} $ is the fundamental ``dark parameter" that is involved in the description of both dark energy and dark matter. \par The mass of galaxies and larger structures involves substantial invisible mass, which must be described with the dark parameter $ a_{0} $. Since the gravitating mass of larger structures is proportional to that of small structures, a scaling relation for $ a_{0} $ must be established. It must scale from the smallest to the largest distances in the Universe. The Planck length $l_{p}=\sqrt{\frac{G\hbar}{c^{3}}}=1.616\times10^{-35}  $m is considered to be the smallest physically relevant length and the size of galaxy clusters is on the higher end of measurable lengths. In the solar system, the measurement of any physical entity is in an environment with a gravitating mass $ M\geq M_{\odot} $ where $ M_{\odot} $ is the mass of the Sun. For large structures in the Universe, the ratio $ \frac{M}{M_{\odot}} $ is a natural scale to use that represents the gravitation of those structures relative to that of the Sun. The dimension of $ a_{0} $ is $ L^{-1} $ so it is defined to be
\begin{equation}\label{a0zeta}
	 a_{0}:=\frac{\zeta Ml_{p}}{M_{\odot}}>0
\end{equation}where $ \zeta $ is a parameter of dimension $ L^{-2} $. Since $ a_{0} $ is the fundamental dark parameter, it must be determined from a known invisible mass halo such as that for galaxy NGC 3198. It has an invisible mass of $ 8.6\times10^{11}M_{\odot} $ and a rotation curve velocity of 150 km/s as discussed in Section 4.5. From (\ref{a0zeta}) and (\ref{TF}), it follows that 
\begin{equation}\label{zeta}
	\zeta=3.55\times10^{-6}\text{m}^{-2}.
\end{equation}
It will be shown that this value of $ \zeta$ leads to reasonable corrections to the GR results in both the solar system and galactic clusters from terms involving $ a_{0}$ and $ b $. The scaling relation (\ref{a0zeta}) generates an estimate for the rotational speed $v$ in (\ref{TF}) when it is not known from observations. It is interesting that the parameter $ b $ scales naturally with $ \sqrt{\frac{GM}{c^{2}}} $.\par  Since (\ref{b}) provides a choice of signs for the parameter $ b $ independent of the sign of the other parameters that define its value, that can lead to unphysical results from the wrong choice of sign. Calculations such as the extended Newtonian force that involves the gradient of the weak field potential $ \phi $ and therefore $e^{-\lambda}  $ introduce a negative sign. In the present epoch, it is apparent that $ b $ must be positive for those calculations. Otherwise, $ b $ is chosen to be negative when the gradient of $ e^{-\lambda} $ is not involved, such as calculations of the Shapiro time delay, gravitational lensing, and the pericentre advance of celestial objects.
\par The Tully-Fisher relation is derived in the next subsection.
\subsection{The radial gravitational force and galactic rotation curves in MGR}
The radial gravitational force on an object of mass m can now be calculated from (\ref{wow}). Using the weak-field relationship of the Newtonian potential $ \phi $ to $ g_{00} $, $	\phi=\frac{c^{2}}{2}(e^{\nu}-1)$, the modified Newtonian force is
\begin{equation}\label{Newt}
	F_{r}=-\frac{GMm}{r^{2}}-\frac{\mid b\mid mc^{2}}{r}+\frac{a_{0}mc^{2}}{2}
\end{equation} where M represents the ordinary mass of the galaxy.\par The correction terms to the Newtonian force come from the non-zero components of the line element field in the energy-momentum tensor $ \varPhi_{\alpha\beta}$. The second term is gravitationally attractive and represents the correction from invisible mass. That term provides the additional gravitational attraction that is missing in GR. Since the parameter $ b $ has two signs, its absolute value is invoked to ensure that term is negative. \par The third term is positive and repulsive with $ a_{0}>0 $ as defined. This describes the repulsive dark energy force in the present accelerating epoch. However, a decelerating epoch has been observed by Riess et al. \cite{Riess}. They used the Hubble telescope to provide the first conclusive evidence for cosmic deceleration that preceded the current epoch of cosmic acceleration.\par Assuming a circular orbit about a point mass, it follows that the orbital velocity of a star rotating in the galaxy satisfies
\begin{equation}\label{V}
	v^{2}(r)=\frac{GM(r)}{r}+\mid b\mid c^{2}-\frac{a_{0}c^{2}}{2}r
\end{equation} where $ M(r) $ is the ordinary mass of the galaxy interior to a fixed radius of $ r $. Equation (\ref{V}) demands an upper limit to r describing a large but finite galaxy.  
\par Because $ a_{0}\neq0$ and scales with the ordinary mass, it is possible for the Newtonian force to balance the dark energy force in a certain region of the galaxy:
\begin{equation}\label{Fd}
		\frac{GM(r)}{r}=\frac{a_{0}c^{2}}{2}r.
\end{equation}
 Then,
\begin{equation}\label{v2b}
	v^{2}=\mid b\mid c^{2}
\end{equation} describes a specific class of galaxies with a flat (slope zero) orbital rotation curve. From (\ref{b}) and (\ref{v2b}), we obtain the exact power four Tully-Fisher relation
\begin{equation}\label{TF}
	v^{4}=\upsilon GM(r)c^{2} a_{0}=GM(r)_{DM}c^{2} a_{0},
\end{equation} which shows that the power four pertains to a pure dark matter region of a galaxy.
\par Using the MOND interpolation function \cite{Milg} $ \mu(\frac{a}{A_{0}})=\frac{a}{A_{0}} $ for $ a<<A_{0} $, the Tully-Fisher relation in MOND is evident by setting 
\begin{equation}\label{mond}
	c^{2}\mid a_{0}\mid:=A_{0}.	
\end{equation} The fundamental acceleration in MOND has the value $ A_{0}\approx1.2\times10^{-10}\,m\,s^{-2}$, which is taken as a fundamental constant. However, as observed in Ref. \cite{MONDA0}, a good MOND fit for galaxy DDO 154 is the small value of $A_{0} = 0.68\times10^{-10}\pm 0.02$, nearly half the standard value. It is apparent that scaling in MOND is important, but that is problematic for it since $ A_{0} $ is taken as a universal constant independent of any galactic property. Whereas $ a_{0} $ is not a fixed parameter in MGR; it is determined from the scaled ordinary mass in a particular region of spacetime according to (\ref{a0zeta}). 
\subsection{Modeling galactic rotation curves with MGR}
Non-visible matter in and around spiral galaxies, distributed differently from stars and gas, is well fixed from optical and 21 cm rotation curves that do not show the expected Keplerian fall-off at large radii but remain increasing, flat or only slightly decreasing over their entire observed range \cite{Sal}. It is clear that (\ref{V}) is a simple but interesting equation from which to model the rotation curves of any radially symmetric galaxy. The parameter $ a_{0} $ is calculated from (\ref{a0zeta}) and (\ref{zeta}); the parameter $ b $ is determined from (\ref{b}). If the rotation curve is flat, there will be a cut-off radius $r_{c}$ determined from observational data above which the visible mass does not significantly contribute to the invisible mass halo, and there is a radius $r_{v}$ that describes the observable outer limit of the dark matter halo. Then (\ref{v2b}) applies in that region and the gravitating dark matter mass responsible for the flat rotation curve at radius $ r $ obtained from (\ref{b}) and (\ref{Fd}) is
\begin{equation}\label{M}
	M(r)_{DM}=\frac{\mid b\mid c^{2}r}{\sqrt{2}G}\;\;\; r_{c}\leq r\leq r_{v}.
\end{equation}The dark matter parameter $ b $ is fundamental to this expression. To obtain an estimate for the mass of the dark matter halo in a galaxy, the mass calculated from (\ref{M}) up to $r_{c}$, which contains mainly baryonic mass, is subtracted from the mass calculated from (\ref{M}) at  $r_{v}$, typically taken as the virial radius.\par Equation (\ref{V}) can be simply expressed as
\begin{equation}\label{Dr}
	v^{2}=\mid b\mid c^{2}+r\Delta (r) 	
\end{equation} where $ \Delta(r):=\frac{GM(r)}{r^{2}}-\frac{a_{0}c^{2}}{2} $ is a function of $ r $ that represents the difference between the Newtonian and dark energy forces. In a region of the galaxy $ r_{c}\leq r\leq r_{v}$ where these forces are slightly different, (\ref{Dr}) describes the linear rise or fall of the rotation curve.  
\subsection{Baryonic galaxies} 
In the far finite extent of the galaxy, the sum of the terms in (\ref{V}) differs from the $ v^{2}\sim\frac{1}{r} $ Keplerian dependence typically assumed in most models. However, if both correction terms from the line element field are small for a particular galaxy, it will exhibit typical Keplerian behaviour, which is rarely true but has been observed for galaxy NGC 1052-DF2. As stated in Ref. \cite{vanD}, it is extremely deficient in dark matter, and a good candidate for a ``baryonic galaxy" with no dark matter at all. Moreover, the observed circular speed profile of the ultra-diffuse galaxy (UDG) AGC 114905 was explained in Ref. \cite{Pina} almost entirely  with the contribution of the baryons alone, with little room for dark matter within the observed outermost radius of $ R\sim 10 $\,kpc. 
\par However, there is \emph{always} dark matter involved in the description of any galactic rotation curve because the line element vector field exists at every point in the spacetime, and $ a_{0} $ and $ b $ never vanish. Using the baryonic mass of 1.4$ \times10^{9}M\odot $ in Ref. \cite{Pina}, $ a_{0}=8.03\times10^{-32}\,m^{-1} $ and $ b=-4.46\times10^{-10} $ with a small value of $\upsilon=1.2$. The Newtonian term in (\ref{V}) dominates the dark terms by a factor of $ \approx $15 at the observed outermost radius of 10\,kpc: $ v^{2}=(6.02\times10^{8}+4.01\times10^{7}-1.11\times10^{6})\,m^{2}\,s^{-2} $. \par The importance of not understimating dark matter in galaxies like AGC 114905 was emphasized in Ref. \cite{Sell} with regard to the stability of the galaxies. They used a NFW profile describing dark matter with a mass of 1.40$ \times10^{9}\,M\odot$ out to 10\,kpc and a total baryonic mass of 1.14$ \times10^{9}\,M\odot $ interior to that radius, so $\upsilon=1.27$ in that region.The Newtonian term in (\ref{V}) is $4.90\times10^{8}\,m^{2}\,s^{-2}$. It dominates the dark terms again by $ \approx $15. While UDG galaxies may appear to be mainly baryonic with a Keplerian behaviour, dark matter and dark energy are always present and cannot be dismissed. 
\subsection{Galaxy NGC 3198}
It is interesting to apply (\ref{Dr}) with $ \Delta=0 $ to galaxy NGC 3198, which is well known for its challenging invisible mass halo. This galaxy exhibits a flat rotation curve with a circular velocity of $v=150$ km/s $\pm 3\% $ over a distance of approximately 31\,kpc from 17\,kpc to 48\,kpc. The virial radius $ r_{vir} $ was modeled \cite{Kar} with the NFW dark matter profile to be 248.9\,kpc with a corresponding virial mass of $M_{vir}=8.9\times10^{11}\,M_{\odot}$. This includes a visible disk mass $ M_{disk}=2.63\times10^{10}\,M_{\odot} $ calculated from $ M_{disk}=1.1\varUpsilon^{3.6}_{*}v^{2}_{d}R_{1}/G $ with $ \varUpsilon^{3.6}_{*}=0.79 $; $ R_{1}=48 $\,kpc is the outermost disk radius and $ v_d=52.2 $ from Table 1 of Ref. \cite{Kar} . The mass of the dark matter halo is then $8.6\times10^{11}\,M_{\odot}$.\par 
From (\ref{M}) and (\ref{v2b}) with $ v=150\,\text{km}\,\text{s}^{-1}$ and $r=r_{vir} $, $ M(r)=9.21\times10^{11}\,M_{\odot} $. As stated in Ref. \cite{Kar}, the dark matter halo extending from 17\,kpc is virtually free from the uncertainty in the actual value of the visible disk mass. The contribution to $M(r)$ up to 17\,kpc is $ 6.28\times10^{10}\,M_{\odot} $ so the MGR estimate for the dark matter halo is $8.6\times10^{11}\,M_{\odot}$. This result is in perfect agreement with that of Karukes et al. \cite{Kar} under the assumption that the rotational velocity is relatively constant out to the virial radius of 249\,kpc. 
\subsection{The Bullet Cluster}
The Bullet Cluster 1E0657-56 represents the collision of two galaxy clusters. Each cluster consists of galaxies and intracluster gas, with the gas comprising the majority of the ordinary matter. Although the individual galaxies within the clusters essentially passed right through one another, the intracluster gas within each cluster collided. The ram pressure heated and ionized the gas, which created a plasma that produced observable X-rays. Using GR, the gravitational lensing of the Bullet Cluster indicated the majority of the mass is alongside and local to the galaxies, not with the plasma. This observation \cite{Clowe} suggests there exists additional invisible mass that is essentially collisionless---dark matter.  \par GR describes the gravity associated with the plasma and galaxies from the Einstein equation. However, it cannot explain the dark matter in the Bullet Cluster. Whereas MGR includes gravitational energy-momentum in the Einstein equation with the tensor $ \varPhi_{\alpha\beta} $. It affects the gravitational lensing of the Bullet Cluster calculated from its gravitating mass. Normally before a collision, the dark matter appears around the ordinary matter as a spherical halo, as with galaxy NGC 3198. However, the collision of the intracluster gas in the Bullet Cluster causes a displacement of the ordinary matter from the dark matter and the lensing shows the spherical dark matter halos alongside the plasma instead of around it.\par The gravitating masses for the main and Bullet cluster, $ 2.5\times10^{14}\,M_{\odot} $ and $ 2.0\times10^{14}\,M_{\odot} $ respectively, were obtained in Ref. \cite{Para} as a consequence of the virial theorem. The virial mass $ M_{vir}$ is related to the virial radius $ r_{200} $ by $ M_{vir}=\frac{4\pi}{3}r_{200}^{3}200\varrho_{c} $ where $\varrho_{c}  $ is the critical density $9.3\times10^{-27}\frac{kg}{m^{3}} $. Given the virial mass, the virial radius can be calculated and used as the impact parameter in (\ref{lens}). The gravitational lensing in MGR for the main cluster is calculated to be 67.04 arcsec. The GR term contributes only 7.65 arcsec to that result because GR cannot describe dark matter. The Bullet cluster lenses 50.71 arcsec with a contribution of 6.58 arcsec from the GR term. Thus, ordinary matter accounts for 11.4\% and 13.0\% of the total gravitating mass of the main and Bullet clusters, respectively, as compared to 15\% for each cluster according to the $ \Lambda $CDM model. 
\subsection{The geometrical and particle descriptions of dark matter} To this point in the discussion of MGR, dark matter is described geometrically in terms of the line element vector $ X^{\beta} $ that is dual to the covector $ X_{\alpha} $  by $ g^{\alpha\beta} $. The line element field is a geometrical entity, which is defined as a pair of equal and opposite regular vectors $ (\bm{X},-\bm{X}) $ at each point on the Lorentzian manifold. One of the pair of vectors is an integral part of the Einstein equation, which contains $ \varPhi_{\alpha\beta}=\frac{1}{2}\pounds_{X}g_{\alpha\beta}+\pounds_{X}(u_{\alpha} u_{\beta}). $\par  That geometrical description seems at odds with the typical explanation of dark matter as an undiscovered invisible particle with an unknown spin. Scalar, fermion, and vector boson dark matter coupled to gravity has been widely studied in the literature, although non-minimally coupled vector boson dark matter has been investigated the least \cite{Bar} (and references therein).\par However, there is no tension between both descriptions of dark matter.
First, it is important to prove that a spin-1 vector boson and a pair of spin-1/2 fermions can be described by pure wave equations of the Klein-Gordon (KG) type. That follows from the discussion in Ref. \cite{NQ} of the asymmetric wave equation for the vector field $X^{\beta}$
\begin{equation}\label{asymKG}
	\nabla_{\alpha}\nabla^{\alpha}X^{\beta}=k^{2}X^{\beta} 
\end{equation}where $ k=\frac{m_{0}c}{\bar{h}} $ and $ m_{0} $ is the rest mass attributed to a particular spin-1 particle. The Proca equation $ \nabla_{\alpha}K^{\alpha\beta}=k^{2}X^{\beta} $ with the Lorentz constraint $ \nabla_{\alpha}X^{\alpha}=0 $ are traditionally used to describe a spin-1 particle in curved spacetime. However, the Proca equation is not a pure wave equation because covariant derivatives do not commute in curved spacetime. That problem is resolved by symmetrizing (\ref{asymKG}) into its equivalent 
\begin{equation}\label{key}
	\nabla^{\alpha}(\tilde{\Psi}_{\alpha\beta}+K_{\alpha\beta})=2k^{2}X_{\beta}
\end{equation}where $\tilde{\Psi}_{\alpha\beta}=\nabla_{\alpha}X_{\beta}+\nabla_{\beta}X_{\alpha}=\pounds_{X}g_{\alpha\beta}  $ and $ K_{\alpha\beta}=\nabla_{\alpha}X_{\beta}-\nabla_{\beta}X_{\alpha} $. $\nabla_{\alpha}X^{\alpha}=(\Phi-2u^{\alpha}u^{\beta}\nabla_{\alpha}X_{\beta})=-\Phi\neq0$ replaces the Lorentz constraint, and  $\tilde{\Psi}_{\alpha\beta}  $ is divergenceless if and only if $k=0$ so that Maxwell's equation holds in curved spacetime. Thus, (\ref{asymKG}) describes a spin-1 vector boson with a pure KG wave equation. Since a spin-1 vector is equivalent to an outer product involving a spin-1/2 Dirac spinor $ \Psi $ and its Hermitian conjugate $ \Psi^{\dagger} $ by the relation $ X^{\beta}=\Psi^{\dagger}\gamma^{0}\gamma^{\beta}\Psi $ where $ \gamma^{\beta} $ are the Dirac gamma matrices, the KG wave equation (\ref{asymKG}) holds for a spin-1 vector boson and a pair of spin-1/2 fermions.\par 
Second, particles that are described by regular vector fields $X^{\beta}$ obey the spin-1 KG wave equation in curved spacetime
\begin{equation}\label{KG1}
	\nabla_{\alpha}\nabla^{\alpha}u^{\beta}=k^{2}u^{\beta} 
\end{equation} obtained from the action in Appendix C: $S^{1}=-\frac{1}{2}\int[\nabla_{\alpha}X_{\beta}\nabla^{\alpha}X^{\beta}+k^{2}X_{\beta}X^{\beta}]\sqrt{-g}d^{4}x $ where $ u^{\beta} $ is collinear to an arbitrary regular vector $X^{\beta}$ whose magnitude satisfies $\square f=0$. The geometry of spacetime and the quantum nature of matter are linked together by the unit line element covectors that belong to both (\ref{gab}) and the dual of (\ref{KG1}), $ \square u_{\alpha}=k^{2}u_{\alpha} $. 
\par This result applies to both dark matter and ordinary matter because $X^{\beta}$ is either a regular vector that is not restricted to the subspace of the line element field, or it is one the pair of regular vectors in the line element field $(X^{\beta},-X^{\beta})$. $\varPhi_{\alpha\beta}$ introduces the regular line element vector $X^{\beta}$ into MGR, which together with its collinear unit vector $u^{\beta}$, are taken to describe dark matter. Since the unit vector that is defined by (\ref{unu}) is generated from line element vectors and their covariant derivatives in Appendix B, (\ref{unu}) pertains solely to dark matter. Then given $ \square u^{\beta}=k^{2}u^{\beta}$, $\square X^{\beta}=k^{2}X^{\beta}$ if and only if $\square f=0$ and $\nabla_{\alpha}f=\Phi u_{\alpha}$. The magnitude of vectors describing dark matter is endowed with a homogeneous wave characteristic but is constrained by the vanishing of (\ref{wf}): $\pounds_{u}\Phi=0  $. \par If $X_{\mu}$ is a regular covector that is not restricted to the subspace of the line element field, variation of the action $S^{1}$ with respect to it and its independent covariant derivative $\nabla_{\nu}X_{\mu}$ yields $ \square X^{\beta}=k^{2}X^{\beta}$, which holds for ordinary matter. Thus, the variation with respect to the magnitude of $X^{\beta}$ in Appendix C generates $\square u^{\beta}=k^{2}u^{\beta} $ for both dark matter and ordinary matter provided the condition $ \square f=0 $ holds. However, for ordinary matter, $ \square X^{\beta}=k^{2}X^{\beta}$ does not follow directly from $ \square u^{\beta}=k^{2}u^{\beta}$, and conversely, because  the unit vector that is defined by (\ref{unu}) pertains solely to dark matter. 
\par Both ordinary matter and dark matter can be related to a given unit vector. Thus, $X_{o}^{\beta}=f_{o}u^{\beta}$ pertains to ordinary matter and $X^{\beta}=fu^{\beta}$ represents dark matter. $f$ is independent of $f_{0}$, and both $\square f$ and $\square f_{o}$ vanish.
\par Lagrangians for interactions between dark matter and ordinary matter particles and the self-interaction of dark matter particles are now investigated. The Lagrangian in the action $ S^{1} $ is of dimension $ L^{-4} $. The valid interaction Lagrangians of dimension $ L^{-4} $ with a dimensionless coupling constant $\lambda_{n}$ are:
\begin{enumerate}
	\item $L^{X_{o}X_{o}}_{dm}=-\lambda_{0} X_{o}^{\alpha}X_{o}^{\beta}\nabla_{\alpha}X_{\beta} $, which describes an interaction of ordinary matter with dark matter. It follows that $L^{X_{o}X_{o}}_{dm}=\lambda_{0}f_{o}^{2}u^{\alpha}\nabla_{\alpha}f=\lambda_{0}f_{o}^{2}u^{\alpha}u_{\alpha}\Phi=-\lambda_{0}f_{o}^{2}\Phi  $. Provided $ X^{\beta} $ is not a Killing vector, $ \Phi $ does not vanish, and dark matter can locally couple to ordinary matter.\par  Writing $\varPhi_{\alpha\beta}$ in terms of $f$, $u_{\alpha}$ and their derivatives allows $\Phi$ to be expressed as $\Phi=-u^{\alpha}\nabla_{\alpha}f$ from which the action $S_{dm}^{X_{o}X_{o}}=\lambda_{0}\int f^{2}_{o}u^{\alpha}\nabla_{\alpha}f\sqrt{-g}d^{4}x=-2\lambda_{0}\int ff_{o}u^{\alpha}\nabla_{\alpha}f_{o}\sqrt{-g}d^{4}x$ follows. Variation with respect to $f$ demands $u^{\alpha}\nabla_{\alpha}f_{o}=\pounds_{u}f_{o}=0$ so $f_{o}$ must be constant along the flow of \textbf{u}. Variation of the action $S_{om}^{XX}=-\lambda_{0}\int X^{\alpha}X^{\beta}\nabla_{\alpha}X_{o\beta}\sqrt{-g}d^{4}x=\lambda_{0}\int f^{2}u^{\alpha}\nabla_{\alpha}f_{o}\sqrt{-g}d^{4}x$ with respect to $f$ generates the same result: $\pounds_{u}f_{o}=0 $.
	\par $\nabla_{\alpha}X^{\alpha}=-\Phi\neq0 $ replaces the Lorentz constraint $ \nabla_{\alpha}X^{\alpha}=0 $. Therefore, $ \Phi $ is expected to be small, which is now demonstrated. In a matter-free region of spacetime, $ \Phi=R $. In the spherical metric (\ref{wow}), $\Phi=\frac{6a_{0}}{r}-\frac{2b(3+2\ln r)}{r^{2}}  $. Near the Earth $ 1.5\times10^{11}\,m $ from the gravitating mass of the Sun, $\Phi=1.42\times10^{-39}\,m^{-2} $. However, it is possible for $\lambda_{0}$ to be large enough so that $\lambda_{0}f_{o}^{2}\Phi$ is measurable. For example, if the average local value of $f_{o}$ is approximately unity, it obeys the condition $\pounds_{u}f_{o}=0$ so  $\lambda_{0}\approx a=\frac{c^{3}}{16\pi G}$ and the interaction has the value $ 1.13\times10^{-5} $. Dark matter particles in the presence of ordinary matter may be detected near the Earth. However, since attempts to detect dark matter near or in the Earth have not yet been confirmed, it is more likely that $\lambda_{0}f_{o}^{2}$ is not large.
	\item  $L^{dm}_{om}=-\lambda_{1} \nabla^{\alpha}X^{\beta}\nabla_{\alpha}X_{o\beta }=\lambda_{1}(\nabla^{\alpha}f\nabla_{\alpha}f_{o}-ff_{o}\nabla^{\alpha}u^{\beta}\nabla_{\alpha}u_{\beta})$ is another interaction between dark and ordinary matter. Integration by parts gives $S^{dm}_{om}=\lambda_{1}\int(\nabla^{\alpha}f\nabla_{\alpha}f_{o}+ff_{o}u^{\beta}\square u_{\beta})\sqrt{-g}d^{4}x$, and variation with respect to $ f $ yields $\lambda_{1}k^{2}f_{o}=0$, which requires the ordinary matter to be massless. That means the electromagnetic field can interact with dark matter. However, $\lambda_{1}$ can be so small that the dark matter will appear invisible.\par  Similarly, variation with respect to $f_{o}$ demands $\lambda_{1}k^{2}f=0$ so the dark matter is massless. Dark photons can interact with ordinary matter with a presumably weak coupling constant.
	\item $L^{dm}_{dm}=-\lambda_{2} \nabla^{\alpha}X^{\beta}\nabla_{\alpha}X_{\beta}=-\lambda_{2}(-\nabla^{\alpha}f\nabla_{\alpha}f+f^{2}\nabla^{\alpha}u^{\beta}\nabla_{\alpha}u_{\beta})$ is a self-interaction of dark matter. Integration by parts yields $S^{dm}_{dm}=\lambda_{2}\int(\nabla^{\alpha}f\nabla_{\alpha}f+f^{2}u^{\beta}\square u_{\beta})\sqrt{-g}d^{4}x$. Variation with respect to $ f $ generates $\lambda_{2}k^{2}f=0 $ so $k=0$. Dark photons interact with themselves through this Lagrangian. However, they must obey the linear wave equation (\ref{asymKG}), which forbids a first-order self-interaction. A minute second-order effect with self-interacting photons has been experimentally verified \cite{burke}. Dark photons may have a similar property, but that has yet to be proven.
	\item The other possible self-interaction of dark matter is $L_{dm}^{XX}=X^{\alpha}X^{\beta}\nabla_{\alpha}X_{\beta}=-f^{2}u^{\alpha}\nabla_{\alpha}f$. Variation of the action for this self-interaction with respect to $f$ generates a vanishing integrand, which implies the Lagrangian is invalid.
	\par Hence, there is no valid Lagrangian that permits a first-order self-interaction of dark matter particles.
\end{enumerate}
Dark matter is gravitationally attractive. That follows from the fact that a positive-definite Riemannian metric $g^{+}_{\alpha\beta}$ exists on any paracompact manifold. It fundamentally determines the curvature of spacetime that is attributed to gravitational attraction. From (\ref{gab}), $g^{+}_{\alpha\beta}=g_{\alpha\beta}+2X_{\alpha}X_{\beta}/|g^{\alpha\beta}X_{\alpha}X_{\beta}|$ involves a pair of line element covectors that adds to the geometry of the Riemannian metric. Moreover, since they describe dark matter particles, the pair of covectors in $g^{+}_{\alpha\beta}$ is gravitationally attractive, as is each particle in the pair. Dark matter persists when ordinary matter is absent because the line element covectors always exist at every point in a Lorentzian spacetime.
\par In summary, dark matter can be explained as a spin-1 unit vector boson or as an equivalent pair of Hermitian spin-1/2 fermions. Dark matter is gravitationally attractive and interacts weakly with ordinary matter. More specifically, the electromagnetic field interacts with dark matter so weakly that it appears invisible. Similarly, dark photons are expected to interact extremely weakly with ordinary matter. There is no first-order self-interaction of dark matter particles. Thus, a massive neutral vector boson and a pair of neutral neutrinos with a presently unknown mass are viable candidates for a particle theory of dark matter. There is no conflict between the geometrical and particle descriptions of dark matter because the unit line element covectors of the Lorentzian metric satisfy the Klein-Gordon wave equation that describes spin-1 particles. In the absence of ordinary matter, dark matter is still present and curves spacetime. GR cannot describe dark matter because the line element field has been ignored.\par Although there is both a geometrical and a quantum theoretical description of dark matter in terms of the line element field vectors, the geometrical description of dark matter is fundamental to gravity. The gravitational force in a metric theory of gravity is entirely determined by the connection on the Lorentzian manifold. The symmetric Levi-Civita connection is constructed from the Lorentzian metric and its first derivatives. The metric is a solution to the Einstein equation, which depends on the line element covectors in $ \varPhi_{\alpha\beta} $ and the matter energy-momentum tensor. Hence, gravitational forces in MGR are completely determined by the curvature of spacetime resulting from all sources of matter and energy. Whereas GR geometrically determines the force of gravity from ordinary matter, it cannot describe the gravitational force from dark matter without introducing a dark matter profile that represents the invisible matter.
\section{Tests of GR compared to MGR}
The three classic tests of general relativity are the following: the perihelion precession of Mercury, gravitational lensing and the Shapiro time delay. The changes induced from the line element field components are calculated and compared with the results from GR for each of these tests. The metric will be that of (\ref{g}) with the spherical solution for $ \lambda $ given by (\ref{wow}) with $ c_{1}=-\frac{2GM}{c^{2}} $ and $ c_{2}=1 $. Then, the time-like norm $ g_{\alpha\beta}\frac{dx^{\alpha}}{d\tau}\frac{dx^{\beta}}{d\tau}=-1 $ leads to 
\begin{equation}\label{tlc}
	-e^{-\lambda} c^{2}(\frac{dt}{d\tau})^{2}+e^{\lambda}(\frac{dr}{d\tau})^{2}+r^{2}[(\frac{d\theta}{d\tau})^{2}+sin^{2}\theta (\frac{d\varphi}{d\tau})^{2}]=-1
\end{equation} and a photon has the null arc length
\begin{equation}\label{nc}
	-e^{-\lambda} c^{2}(\frac{dt}{d\tau})^{2}+e^{\lambda}(\frac{dr}{d\tau})^{2}+r^{2}[(\frac{d\theta}{d\tau})^{2}+sin^{2}\theta (\frac{d\varphi}{d\tau})^{2}]=0.
\end{equation}\par  The value of the parameters $ a_{0} $ and $ b $ with a gravitating mass $ M $ are obtained from (\ref{a0zeta})  and  (\ref{b}) respectively with $ \zeta=3.55\times 10^{-6} $ from (\ref{zeta}).
\subsection{Pericenter advance of celestial objects}
A test mass moves in a gravitational field free of external forces according to the geodesic equation
\begin{equation}\label{geo}
	\frac{d^{2}x^{\lambda}}{d\tau^{2}}+\Gamma^{\lambda}_{\alpha\beta}\frac{dx^{\alpha}}{d\tau}\frac{dx^{\beta}}{d\tau}=0
\end{equation}and satisfies (\ref{tlc}). From the spherically symmetric time independent metric, it is well known \cite{Ciu} that in the equatorial plane, $ \theta=\frac{\pi}{2} $, $ \frac{d\theta}{d\tau}=0 $ and the constants of motion are
\begin{equation}\label{L}
	L:=r^{2}\frac{d\varphi}{d\tau}
\end{equation}and
\begin{equation}\label{K}
	K:=e^{-\lambda}\frac{dt}{d\tau}.
\end{equation}The $ x^{1} $ geodesic yields
\begin{equation}\label{x1g}
	\frac{d^{2}r}{d\tau^{2}}-\frac{\lambda^{\prime}c^{2}K^{2}}{2}+\frac{\lambda^{\prime}}{2}(\frac{dr}{d\tau})^{2}-\frac{e^{-\lambda}L^{2}}{r^{3}}=0
\end{equation}
and the $ x^{3}$ geodesic gives 
\begin{equation}\label{x3g}
	\frac{d^{2}\varphi}{d\tau^{2}}=-\frac{2}{r}\frac{d\varphi}{d\tau}\frac{dr}{d\tau}
\end{equation}from which
\begin{equation}\label{key}
	\frac{d^{2}r}{d\tau^{2}}=\frac{d^{2}r}{d\varphi^{2}}\frac{L^{2}}{r^{4}}-\frac{2L^{2}}{r^{5}}(\frac{dr}{d\varphi})^{2}.  
\end{equation}By introducing the variable $ y:=\frac{1}{r} $, it follows from (\ref{x1g}) that 
\begin{equation}\label{Peq}
	\begin{split}
		\frac{d^{2}y}{d\varphi^{2}}+y-(\frac{a_{0}}{2}+\frac{GM}{c^{2}L^{2}})=by(1+2\ln y)+\frac{3GMy^{2}}{c^{2}}+\frac{b}{L^{2}y}-\frac{a_{0}}{2L^{2}y^{2}}.
	\end{split}
\end{equation} The homogeneous equation from (\ref{Peq})
\begin{equation}\label{}
	\frac{d^{2}y_{0}}{d\varphi^{2}}+y_{0}-A_{1}=0
\end{equation} has the solution
\begin{equation}\label{key}
	y_{0}=A_{1}(1+e\cos\varphi)
\end{equation}where $ A_{1}:=\frac{a_{0}}{2}+\frac{GM}{c^{2}L^{2}} $, $ e:=\frac{c_{4}}{A_{1}} $ and $ c_{4} $ is an arbitrary constant. This reduces to the Newtonian solution for $ a_{0}=0 $. A solution to (\ref{Peq}) is sought with a small perturbation $ y_{1} $ around $ y_{0} $, which means $ \ln y\approx\ln y_{0}\leq\ln A_{1}(1+e) $ for a maximum contribution from $ e\cos\varphi $ of $e$. By defining $ B:=1-b(1+2\ln A_{1}(1+e)) $ as a parameter independent of $ \varphi $, equation (\ref{Peq}) is well approximated by 
\begin{equation}\label{Peq1}
	\frac{d^{2}y}{d\varphi^{2}}+By-(\frac{a_{0}}{2}+\frac{GM}{c^{2}L^{2}})=\frac{3GMy^{2}}{c^{2}}+\frac{b}{L^{2}y}-\frac{a_{0}}{2L^{2}y^{2}}.
\end{equation}
With $ y=y_{0}+y_{1} $ where $ y_{1}<<y_{0} $ and the change of variable $ \chi=\sqrt{B}\varphi $, (\ref{Peq1}) is approximated by
\begin{equation}\label{Peq2}
	\frac{d^{2}y_{1}}{d\chi^{2}}+y_{1}=\frac{3GMy_{0}^{2}}{Bc^{2}}+\frac{b}{BL^{2}y_{0}}-\frac{a_{0}}{2BL^{2}y_{0}^{2}},
\end{equation}which is expanded to second order in $ e $ to give
\begin{equation}\label{y1p}
	\frac{d^{2}y_{1}}{d\chi^{2}}+y_{1}=A_{2}+A_{3}e\cos\chi+A_{4}e^{2}\cos^{2}\chi
\end{equation}
where $ A_{2}:=\frac{3GMA_{1}^{2}}{Bc^{2}}+\frac{b}{BL^{2}A_{1}}-\frac{a_{0}}{2BL^{2}A_{1}^{2}} $, $ A_{3}:=\frac{6GMA_{1}^{2}}{Bc^{2}}-\frac{b}{BL^{2}A_{1}}+\frac{a_{0}}{BL^{2}A_{1}^{2}} $ and $ A_{4}:=\frac{3GMA_{1}^{2}}{Bc^{2}}+\frac{b}{BL^{2}A_{1}}-\frac{3a_{0}}{2BL^{2}A_{1}^{2}} $. Equation (\ref{y1p}) has the solution
\begin{equation}\label{key}
	\begin{split}
		y_{1}=A_{2}+b_{3}\sin\chi+b_{4}\cos\chi+\frac{A_{3}e\chi\sin\chi}{2}-\frac{A_{4}e^{2}\cos2\chi}{6}+\frac{A_{4}e^{2}}{2}
	\end{split}
\end{equation}where $ b_{3} $ and $ b_{4} $ are arbitrary constants. The term containing $ \varphi\sin\varphi $ can grow in $ \varphi $ and is most interesting. The other terms are small and are discarded so that
\begin{equation}\label{key}
	y\approx \frac{A_{1}}{B}[1+e(\cos\chi+\frac{BA_{3}\chi\sin\chi} {2A_{1}})]. 
\end{equation}Since $ \alpha:=\frac{BA_{3}}{2A_{1}}<<1 $, it follows from the identity $ \cos[\chi(1-\alpha)]\approx\cos\chi+\alpha\chi\sin\chi $ that
\begin{equation}\label{key}
	y\approx \frac{A_{1}}{B}[1+e\cos(\chi(1-\alpha))]
\end{equation}The pericenter occurs at the minimum value of $ r $. This requires $ \cos[\chi(1-\alpha)]=1 $ which means $ \chi=\frac{2\pi n}{1-\alpha}\approx2\pi n(1+\alpha) $ for $ n $ orbits. The advance of the pericenter from each orbit is therefore $ \Delta\chi=2\pi\alpha $ so
\begin{equation}\label{peri}
	\begin{split}
		\Delta\varphi=\frac{6\pi A_{1} GM}{\sqrt{B}c^{2}}-\frac{\pi b}{\sqrt{B}L^{2}A_{1}^{2}}+\frac{\pi a_{0}}{\sqrt{B}L^{2}A_{1}^{3}}\\
		=\frac{6\pi GM}{\sqrt{B}c^{2}}(\frac{a_{0}}{2}+\frac{1}{a(1-e^{2})})-\frac{\pi b}{\sqrt{B}L^{2}A_{1}^{2}}+\frac{\pi a_{0}}{\sqrt{B}L^{2}A_{1}^{3}}
	\end{split}
\end{equation}where $ L^{2}=\frac{GMa(1-e^{2})}{c^{2}} $ and $ a $ is the semi-major axis. 
\subsubsection{Precession of the perihelion of Mercury}
In the solar system with the gravitating mass of the Sun containing some dark matter, $M=1.989\times10^{30}$kg and $ a_{0}=5.737\times10^{-41} $ so $ A_{1}\approx\frac{GM}{c^{2}L^{2}}=1.803\times10^{-11} $, $ b=-2.911\times10^{-19} $, and $B=1+1.40\times10^{-17} $, which is unity to seventeen orders of magnitude. The first term equals $7.97\times10^{-37}$rad, which is neglected. The second term is the $\frac{6\pi GM}{c^{2}a(1-e^{2})}  $ expression of GR. For Mercury, $ a=5.7909\times10^{10} $\,m and $ e=0.2056 $ and its perihelion advance calculated from (\ref{peri}), term by term, is
\begin{equation}\label{Dphi}
	\Delta\varphi=5.0201\times10^{-7}+3.4341\times10^{-11}+3.7539\times10^{-22}\;\text{rad}
\end{equation}or 42.98 arc seconds per century, which is the same as the GR result to 4 orders of magnitude.\par The third term in (\ref{Dphi}) represents the dark energy correction, which is negligible in the solar system. However, the second term pertains to dark matter, which competes with the correction of the quadrupole moment of the Sun: $ \Delta\varphi_{quad}=-\frac{3\pi J_{2}R^{2}_{\odot}(3\sin^{2}\iota-1)}{a^{2}(1-e^{2})^{2}}=2.8361\times10^{-10}  $ using the currently accepted value of $ J_{2}=2.0\times10^{-7} $ for the quadrupole moment of the Sun, $ R_{\odot }=6.95997\times10^{8} $\,m and $ \iota=7.00487^{\circ} $. The dark matter term represents a correction of 0.3434 to the quadrupole term, which gives a total perihelion advance of 43.01 arc sec/century from GR, dark matter and the quadrupole moment of the Sun. That result compares very well with the measured value of $ 43.0115\pm 0.0085$\,arc seconds per century for the perihelion advance of Mercury as reported in Table 1 of Ref. \cite{Solar}.
\subsubsection{Precession of S2 around Sagittarius A*} The star S2 orbiting the black hole Sagittarius A* at the center of the Milky Way galaxy is a precision test of the gravitational field around a massive black hole. The mass of the black hole is estimated \cite{Abu} to be $ 4.251\times10^{6}M_{\odot} $ and the observed semi-major axis of S2 is 125.058 mas or $ 1.543\times10^{14} $\,m using $ R_{0}=8.247 $\,kpc. Since $ a_{0}=2.4387\times10^{-34} $, $ b=-1.2374\times10^{-12} $, $ A_{1}\approx\frac{GM}{c^{2}L^{2}}$ and  
\begin{equation}\label{key}
	\Delta\varphi=3.5269\times10^{-3}+2.0777\times10^{-8}+1.3741\times10^{-16}.
\end{equation}The GR term is dominant over the dark matter correction by 5 orders of magnitude and $\Delta\varphi\approx $12.1 arcmin per revolution. 
\subsection{Gravitational lensing}
From the null arc length (\ref{nc}) in the equatorial plane with $ \frac{d\theta}{d\tau}=0 $, 
\begin{equation}\label{drdtau}
	(\frac{dr}{d\tau})^{2}+\frac{L^{2}}{r^{2}e^{\lambda}}=c^{2}K^{2}
\end{equation}using (\ref{L}) and (\ref{K}). Since $\frac{dr}{d\tau}=\frac{dr}{d\varphi}\frac{L}{r^{2}}  $ and $\frac{dr}{d\varphi}=-r^{2}\frac{dy}{d\varphi}$ with $ y:=\frac{1}{r} $, it follows from (\ref{wow}) that 
\begin{equation}\label{key}
	(\frac{dy}{d\varphi})^{2}-a_{0}y-2by^{2}\ln y-\frac{2GM}{c^{2}}y^{3}+y^{2}=\frac{c^{2}K^{2}}{L^{2}}.
\end{equation}Differentiating this equation with respect to $ \varphi $ yields
\begin{equation}\label{Leq}
	\frac{d^{2}y}{d\varphi^{2}}-\frac{a_{0}}{2}+y=by+2by\ln y+\frac{3GMy^{2}}{c^{2}}.
\end{equation}The homogenous equation from (\ref{Leq}) has the solution 
\begin{equation}\label{key}
	y_{0}=\frac{a_{0}}{2}+\frac{\sin\varphi}{I},
\end{equation} which satisfies the boundary conditions $y_{0}(\varphi=0)=\frac{a_{0}}{2}$ and $y_{0}^{\prime}(\varphi=0)=\frac{1}{I}$ where $ I $ is the impact parameter and the prime denotes the first derivative with respect to $ \varphi $. \par A solution to (\ref{Leq}) is sought with a small perturbation $ y_{1} $ around $ y_{0} $, which means $ \ln y\approx\ln y_{0}$. The expression $1-b(1+2\ln(\frac{a_{0}}{2}+\frac{\sin\varphi}{I}))  $ for the solar system with $ I=6.957\times 10^{8} $ and $ b=-2.911\times 10^{-19} $ has the constant value of 1, to 17 decimal places in the interval $ 0\leq\varphi\leq\pi. $ The same expression for a galaxy cluster of mass $ 10^{14}M_{\odot} $ with $ I=1.543\times 10^{22} $ and $ b=-2.911\times 10^{-5} $ generates the value 0.9971 with a variance of 0.06\% over the same interval, and a large galaxy cluster of mass $ 10^{15}M_{\odot} $ with $ I=7.715\times 10^{22} $ and $ b=-2.911\times 10^{-4} $ yields 0.9696 with a variance of 0.37\% over the interval. That expression can then be taken to be essentially independent of $ \varphi $ and is represented by the parameter $ B:=1-b(1+2\ln(\frac{a_{0}}{2}+\frac{1}{I})) $, which includes the maximum value of $ y_{0} $. A change of variables with $ \chi:=\sqrt{B}\varphi $ allows equation (\ref{Leq}) to be expressed as
\begin{equation}\label{yphi}
	\frac{d^{2}y}{d\chi^{2}}-\frac{a_{0}}{2B}+y=\frac{3GMy^{2}}{Bc^{2}},
\end{equation}which has the homogeneous solution $ y_{0}=\frac{a_{0}}{2B}+\frac{\sin\chi}{I} $. With $ y=y_{0}+y_{1} $ where $ y_{1}<<y_{0} $, 
\begin{equation}\label{ychi2}
	\frac{d^{2}y_{1}}{d\chi^{2}}+y_{1}=\frac{3GMy_{0}^{2}}{Bc^{2}},
\end{equation} 
which has the solution
\begin{equation}\label{key}
	\begin{split}
		y_{1}=A-\frac{C\chi\cos\chi}{2}+\frac{D\cos 2\chi}{6}+\frac{D}{2}+c_{6}\sin\chi+c_{7}\cos\chi
	\end{split}
\end{equation}where $ 	A:=\frac{3GMa_{0}^{2}}{4c^{2}B^{3}}$, $C:=\frac{3GMa_{0}}{c^{2}B^{2}I}$, $D:=\frac{3GM}{c^{2}I^{2}B}$
and $ c_{6} $ and $ c_{7} $ arbitrary constants. Applying the boundary conditions $y_{0}(\chi=0)=\frac{a_{0}}{2B}$ and $y_{0}^{\prime}(\chi=0)=\frac{1}{I}$ to $ y=y_{0}+y_{1} $ gives
\begin{equation}\label{c67}
	c_{6}=\frac{C}{2},\enspace c_{7}=-A-\frac{2D}{3}.
\end{equation}\par The deflection angle $ \delta $ is now calculated from $ y(\pi+\delta)=0 $. Using $ \sin(\pi+\delta)\approx-\delta $, $\cos(\pi+\delta)\approx-1  $, $\cos(2\pi+2\delta)\approx1  $, (\ref{c67}) and the definitions of $ A $, $ C $ and $ D $, it follows that
\begin{equation}\label{lens}
	\delta=\frac{4GM}{c^{2}BI}+\frac{3GMa_{0}^{2}I}{2c^{2}B^{3}}+\frac{a_{0}I}{2B}+\frac{3\pi GMa_{0}}{2c^{2}B^{2}}.
\end{equation}
\subsubsection{Gravitational lensing from the Sun}
With $ I=R_{\odot}=6.957\times 10^{8}\,m $ and $ a_{0} $ as given above, the last three terms are 26 to 63 orders of magnitude smaller than the first term, and are insigificant. The correction to the GR term is therefore from $ B $ in $\delta=\frac{4GM}{c^{2}BI}$ where $ B=1+1.156\times10^{-17} $ using the Earth-Sun distance of $ 1.496\times10^{11} $\,m. In the solar system, there is no correction to the GR lensing result to 17 orders of magnitude.
\subsubsection{Gravitational lensing from galactic clusters}
In the solar system, the value of the dark energy parameter $ a_{0} $ is far too small to have any gravitational lensing effect. However, that is not the case for galactic clusters, which have enormous masses and large radii. The parameter $ a_{0} $ in MGR scales with the gravitating mass and provides significant contributions to the gravitational lensing of those entities. \par Galactic clusters typically have a total mass of $ 10^{14}\,M\odot$ to $10^{15}\,M\odot $ and a diameter of $ 1$ to $ 5 $\,Mpc. They can produce multiple images separated by several arc minutes. The gravitational lensing from spherical galactic clusters inline with the source and the observer can be easily calculated from (\ref{lens}). \par The strong lensing galactic cluster SDSS J0900+2234 with an Einstein ring was discussed in Ref. \cite{Wie}. They reported in Table 4 a gravitating mass of $1.48\pm0.715\times 10^{14}\,M_{\odot}$ with a dark matter contribution calculated from the NFW dark matter profile. The redshift of the lens and the source are $z_{l}=0.4890 $ and $z_{s}=2.0325$ respectively, and the Einstein ring had a radius of $\theta_{E}=8.0\pm2.7 $\,arcsec. The lens equation in GR requires $ \hat{\alpha}=\frac{4GM}{c^{2}I}=\frac{D_{s}}{D_{ls}}\theta_{E}\approx\frac{z_{s}}{z_{s}-z_{l}}\theta_{E}=10.5\pm3.6 $, which must be used to compare (\ref{lens}) to the Einstein radius because $\hat{\alpha}  $ is the main part of the first term in (\ref{lens}). $D_{s}  $ and $ D_{ls} $ are the angular diameter distances to the source and from the lens to the source, respectively. In MGR, a radius of $ 500 $\,kpc is used for the impact parameter $ I $. If the lower limit of the gravitating mass $ 0.765\times10^{14}\,M_{\odot} $ is used, $ a_{0}=4.39\times 10^{-27} $, $ b=-2.23\times10^{-5} $ and $ B=0.9978 $. The term by term gravitational lensing from (\ref{lens}) is then: $2.936\times10^{-5}+5.071\times10^{-14}+3.394\times10^{-5}+ 2.346\times10^{-9}\, \text{rad}=13.06 $\,arcsec; whereas the GR term (with no dark matter) would generate the single value of $6.041$\,arcsec. The MGR value of 13.06\,arcsec compares well and is less than the reported upper limit as adjusted of 14.1\,arcsec. However, using the mean gravitating mass of $ 1.48\times10^{14}\,M\odot $ generates a lensing of 25.31\,arcsec, which indicates the gravitating mass is too large relative to the size of the observed Einstein radius. There is some uncertainty in the assumption of a perfect line of sight from the Earth to the cluster, and the mass estimate from the NFW profile for the dark halo of SDSS J0900+2234 may be too large.  \par The gravitational lensing of a large galactic cluster of mass $10^{15}\,M_{\odot}  $ with a radius of $ 2.5  $\,Mpc is calculated to be 487.14\,arcsec in MGR as compared to 15.80\,arcsec from the GR term (with no dark matter) using $ a_{0}=5.74\times 10^{-26} $, $ b=-2.91\times10^{-4} $ and $ B=0.9696 $. Thus, MGR generates the general characteristics of the gravitational lensing of galactic clusters from the dark parameters of the line element vectors.
\subsection{Gravitational time delay}
\subsubsection{Shapiro time delay in the Solar system}
From (\ref{drdtau}), 
\begin{equation}\label{key}
	(\frac{dr}{d\tau})^{2}=c^{2}K^{2}-\frac{L^{2}}{r^{2}e^{\lambda}}, 
\end{equation}which is equivalent to
\begin{equation}\label{key}
	\frac{dr}{dt}=ce^{-\lambda}\sqrt{1-\frac{L^{2}e^{-\lambda}}{c^{2}K^{2}r^{2}}}
\end{equation}by considering only the positive square root. At the closest approach to the sun, $r=I$ and $ \frac{dr}{dt}=0 $ so $ \frac{L^{2}e^{-\lambda}}{c^{2}K^{2}I^{2}}=1 $ and $ \frac{L^{2}}{K^{2}}=c^{2}I^{2}e^{\lambda}\mid_{r=I} $. Then
\begin{equation}\label{key}
	\frac{dr}{dt}=ce^{-\lambda}\sqrt{1-\frac{I^{2}(1-a_{0}r+2b\ln r-\frac{2GM}{c^{2}r})}{r^{2}(1-a_{0}I+2b\ln I-\frac{2GM}{c^{2}I})}}.
\end{equation} Inspection of $ e^{-\lambda} $ for $ r\approx 10^{11} $ shows $ -a_{0}r+2b\ln r-\frac{GM}{c^{2}r} $ to have the approximate values $-1.616\times10^{-30}-7.826\times10^{-18}-1.477\times10^{-8}  $. As a first approximation, the terms involving $ a_{0} $ and $ b $ inside the square root expression are dropped. Then,
\begin{equation}\label{key}
	\frac{dr}{dt}\approx ce^{-\lambda}\sqrt{\frac{r^{2}-I^{2}}{r^{2}}-\frac{2GMI(r-I)}{c^{2}r^{3}}},
\end{equation}which leads to the approximation
\begin{equation}\label{key}
	\frac{dt}{dr}=\frac{r}{c\sqrt{r^{2}-I^{2}}}[1-2b\ln r+\frac{2GM}{c^{2}r}+\frac{GMI}{c^{2}r(r+I)}].
\end{equation}Proper time on Earth is close to the Schwarzschild coordinate time, so the total time of travel of the radar signal is essentially
\begin{equation}\label{key}
	T_{t}=2\int_{I}^{r_{E}}\frac{dt}{dr}+2\int_{I}^{r_{p}}\frac{dt}{dr}
\end{equation}where $ r_{E} $ and $ r_{p} $ are the distances of the Earth and a planet from the Sun, respectfully. Integrating each term gives
\begin{equation}\label{key}
	T_{t}=\frac{2}{c}(\sqrt{r_{E}^{2}-I^{2}}+\sqrt{r_{p}^{2}-I^{2}})+\Delta T
\end{equation}where $ \Delta T $ is the time delay due to the gravitational field:
\begin{equation}\label{key}
	\begin{split}
		\Delta T=\frac{2GM}{c^{3}}(\frac{\sqrt{r_{E}^{2}-I^{2}}}{r_{E}+I}+\frac{\sqrt{r_{p}^{2}-I^{2}}}{r_{p}+I})\\+\frac{4GM}{c^{3}}[\ln\frac{r_{E}+\sqrt{r_{E}^{2}-I^{2}}}{I}+\ln\frac{r_{p}+\sqrt{r_{p}^{2}-I^{2}}}{I}]\\-\frac{4b}{c}[\sqrt{r_{E}^{2}-I^{2}}(\ln r_{E}-1)+\sqrt{r_{p}^{2}-I^{2}}(\ln r_{p}-1)\\-I\tan^{-1}(\frac{I}{\sqrt{r_{E}^{2}-I^{2}}})-I\tan^{-1}(\frac{I}{\sqrt{r_{p}^{2}-I^{2}}})+\pi I]. 
	\end{split}
\end{equation} With $ r_{E}=1.496\times10^{11}m $ for the Earth-Sun distance, $ r_{p}=1.08\times10^{11}m $ for the Venus-Sun distance, $ I=R_{\odot}=6.957\times10^{8}m $, the result calculated from the GR terms is 252$\mu s $. The last term is the correction from the parameter $ b $ of 2.46$\times10^{-8}  \mu s$. \par Thus, time delay corrections to the GR results from the line element vectors in the extended Schwarzschild metric of MGR are very small in the solar system. However, that is not the case for larger structures in the cosmos.
\subsubsection{Galactic Shapiro delay from the Crab pulsar PSR B0531+21} The galactic Shapiro delay is more generally calculated from 
\begin{equation}\label{DeltaT}
	\Delta t:=\Delta t_{GR}+\Delta t_{b}+\Delta t_{a}=-\frac{2}{c^{3}}\int _{source}^{observer}\phi(s)ds
\end{equation}where $ \phi $ is the gravitational potential in the extended Schwarzschild metric
\begin{equation}\label{}
	\phi(r)=-\frac{GM}{r}+bc^{2}\ln(r)-\frac{a_{0}rc^{2}}{2}.
\end{equation}The Crab pulsar has a visible line of sight from the Earth. The static gravitational effect of the Milky Way galaxy acting as a lens on light emitted from the pulsar and observed on Earth is depicted (not to scale) in figure 1:
\begin{figure}[H]
	\begin{tikzpicture}[x=0.5cm,y=0.5cm,z=0.3cm,>=stealth]
		\draw (-.5,2.5) node {I};
		\draw (3.2,3.6) node {r};
		\draw (2.8,5) circle (.1cm);
		\draw(3,5.8) node {Observer};
		\draw (8.5,5) circle (.15cm);
		\draw(8.5,5.8) node {Source};
		\draw (0,0) circle (.35cm);
		\draw(0,-1.1) node {Lens};
		\draw (0,0) -- (0,5) -- (8.660,5);
		\draw (0,0) --(2.887,5);
		\draw [dashed] (0,0) -- (8.660,5);
		\draw (0,0) --(5,5);
		\draw (0,2) arc (90:45:2);
		\node[] at (70:1.5)  {$\beta$};
		\draw (0,3) arc (90:30:3);
		\node[] at (52:2.5)  {$\phi$};
		\draw (0,1) arc (90:60:1);
		\node[] at (38:3.5)  {$\theta$};
		\draw [latex-latex](0,8)--(3,8);
		\draw [latex-latex](3,8)--(8.660,8);
		\node[] at (1.5,7.5)  {$D_{l}$};
		\node[] at (5.5,7.5) {$D_{s}$};
		\draw [latex-latex](0,6.8)--(5,6.8);
		\node[] at (2.5,6.4) {s}; 
		\draw [latex-latex](0,8.5)--(8.660,8.5);
		\node[] at (4.3,8.9) {l}; 
	\end{tikzpicture}
	\caption{Angular arrangement for the gravitational lensing by the Milky Way}
\end{figure}The origin is located at the galactic center, which is the location of the lens. The limits of the integrations are reversed to ensure the time delays are positive. The first term of (\ref{DeltaT}) appears from General Relativity:
\begin{equation}\label{key}
	\Delta T_{GR}=\frac{2GM}{c^{3}}(\int_{0}^{\theta}\frac{d\phi}{\cos\phi}-\int_{0}^{\beta}\frac{d\phi}{\cos\phi})
\end{equation}using $ \cos\phi=\frac{I}{r} $ and $ \frac{ds}{r}=\frac{d\phi}{\cos\phi} $. Since $ \int \frac{d\phi}{\cos\phi}=ln(\sec\phi+\tan\phi) $, the integration gives
\begin{equation}\label{DeltaGR}
	\Delta T_{GR}=\frac{2GM}{c^{3}}\ln(\frac{\sqrt{I^{2}+l^{2}}+l}{\sqrt{I^{2}+D_{l}^{2}}+D_{l}})
\end{equation}where $ \tan\beta=\frac{D_{l}}{I} $ and $ \cos\beta=\frac{I}{\sqrt{I^{2}+D_{l}^{2}}}.  $\par The second term of (\ref{DeltaT}) is the dark matter correction to $ \Delta T $, which involves the parameter $b$. It involves the integrals $\int \frac{d\phi}{\cos^{2}\phi}=\tan\phi  $ and $ \int \frac{\ln\cos\phi }{\cos^{2}\phi}d\phi=-\phi+\tan\phi+\tan\phi\ln\cos\phi $:
\begin{equation}\label{key}
	\begin{split}
		\Delta T_{b}=-\frac{2bI\ln I}{c}(\int_{0}^{\theta}\frac{d\phi}{\cos^{2}\phi}-\int_{0}^{\beta}\frac{d\phi}{\cos^{2}\phi}) +\frac{2bI}{c}(\int_{0}^{\theta}\frac{\ln\cos\phi }{\cos^{2}\phi}d\phi-\int_{0}^{\beta}\frac{\ln\cos\phi }{\cos^{2}\phi}d\phi), 
	\end{split}
\end{equation} which integrates to
\begin{equation}\label{Deltab}
	\begin{split}
		\Delta T_{b}=-\frac{2b\ln I}{c}(l-D_{l})+\frac{2bI}{c}(-\tan^{-1}(\frac{l}{I})+\tan^{-1}(\frac{D_{l}}{I}))\\
		+\frac{2b}{c}(l+l\ln(\frac{I}{\sqrt{I^{2}+l^{2}}})-D_{l}-D_{l}\ln(\frac{I}{\sqrt{I^{2}+D_{l}^{2}}}))\qquad b<0.
	\end{split}
\end{equation}
\par The third term of (\ref{DeltaT}) is proportional to the dark parameter $ a_{0} $: 
\begin{equation}\label{key}
	\Delta T_{a}=\frac{a_{0}I^{2}}{c}(\int_{0}^{\theta}\frac{d\phi}{\cos^{3}\phi}-\int_{0}^{\beta}\frac{d\phi}{\cos^{3}\phi}),
\end{equation}which, using the integral $ \int\frac{d\phi}{\cos^{3}\phi}=\frac{1}{2}\frac{\tan\phi}{\cos\phi}+\frac{1}{2}\ln(\sec\phi+\tan\phi) $ integrates to
\begin{equation}\label{Deltaa0}
	\begin{split}
		\Delta T_{a}=\frac{a_{0}}{2c}(l\sqrt{l^{2}+I^{2}}-D_{l}\sqrt{D_{l}^{2}+I^{2}})+\frac{a_{0}I^{2}}{2c}\ln(\frac{\sqrt{l^{2}+I^{2}}+l}{\sqrt{D_{l}^{2}+I^{2}}+D_{l}})\qquad a_{0}>0.
	\end{split}
\end{equation} 
The impact parameter is estimated from: \newline $
I=r_{G}\sqrt{1-(\sin\delta_{S}\sin\delta_{G}+\cos\delta_{S}\cos\delta_{G}\cos(\beta_{S}-\beta_{G}))^{2}}$ where $ \delta_{S} $ and $ \delta_{G} $ are the declinations of the source and galactic centre respectively, and $ \beta_{S} $ and $ \beta_{G} $ are their right ascensions in the equatorial coordinate system. The galactic centre of the Milky Way (MW) has coordinates $ \delta_{G}=-29^{\circ} 00'28.1'' $, $ \beta_{G}=17^{h}45^{m}40.04^{s} $ and the Crab pulsar is located at $ \delta_{S}=22^{\circ} 00'52.1'' $, $ \beta_{S}=05^{h}34^{m}32^{s} $. With $ r_{G}=8.3 $\,kpc, the impact parameter is estimated to be $ I=1.073 $\,kpc. $ D_{l}=8.3 $\,kpc and $ D_{s}=2.2 $\,kpc.  \par The estimate of $5.4\times10^{10}M_{\odot}$ in \cite{Posti} for the ordinary mass of the Milky Way within 20 kpc is used for $M$, from which $a_{0}=3.729\times10^{-30}$ and $b=-1.892\times10^{-8}\sqrt{\upsilon}$; $\upsilon$ is set at $\upsilon=2.54$ as the dark matter/ordinary matter ratio as determined from \cite{Posti} up to 20 kpc from the galactic center. The Shapiro time delays are then: $\Delta t_{GR}=1.44$d, $\Delta t_{b}=6.19$d, and $\Delta_{a}=2.37\times10^{-3}$d for a total delay of 7.63d. This result can be compared to three calculations of the Shapiro time delay for the Crab pulsar quoted in Ref. \cite{Des} as (5.14-15.42)\,d, 1.98\,d and 3.84\,d. The first result was determined from the sum of the gravitational potentials of the MW and the Crab pulsar. The latter two results used the NFW dark matter profile to determine the additional gravitational delay due to dark matter. The MGR result is larger than the last two noted results commensurate with the fact that the NFW dark matter profile is not valid \cite{Cautun} in the interval [5,30] kpc from the galactic centre of the Milky Way
\section{Cosmological aspects of MGR}
\subsection{MGR vs. $ \Lambda  $CDM}
The standard $ \Lambda$CDM model of cosmology is based on the following assumptions:
\begin{enumerate}
	\item   General Relativity is correct and describes gravity on cosmological scales.
	\item The Universe consists of ordinary matter, radiation including neutrinos, cold dark matter responsible for structure formation and a cosmological constant $ \Lambda $ that represents dark energy, which describes the accelerated expansion of the observable Universe. 
	\item The Cosmological Principle states that the Universe is homogeneous and isotropic at large scales $>$100\,Mpc, which is described by the flat FLRW metric. 
	\item Inflation, a violent accelerated expansion of the Universe, immediately followed the Big Bang. 
\end{enumerate} \par The $ \Lambda$CDM model has been successful in explaining most cosmological observations including the accelerating expansion of the Universe, the power spectrum and statistical properties of the cosmic microwave background anisotropies, the spectrum and statistical properties of large scale structures of the Universe and the observed abundances of different types of light nuclei.\par Relative to the $ \Lambda$CDM model, MGR has, but is not limited to, the following attributes:
\begin{enumerate}
	\item GR is incomplete; it does not contain a symmetric tensor that represents the energy-momentum of the gravitational field. MGR completes GR by including $\varPhi_{\alpha\beta}$, which leaves the Einstein equation intact with the total energy-momentum tensor $ T_{\alpha\beta} $ given by (\ref{Tab}).
	\item The matter tensor $\tilde{T}_{\alpha\beta}$ represents ordinary matter, radiation and neutrinos. $ \varPhi_{\alpha\beta} $ introduces the line element covectors into the Einstein equation that geometrically describe dark matter. The cosmological constant is dynamically replaced by the scalar $\Phi$. $\Phi>-2\varPhi_{00}$ is gravitationally repulsive and locally describes dark energy. $\Phi$ satisfies (\ref{intPhi}), which keeps a Lorentzian spacetime in balance; it cannot rip apart or contract to oblivion. 
	\item The Cosmological Principle applies to MGR in a limited sense as compared to the $ \Lambda$CDM model. $ \varPhi_{00} $ and $ \Phi $ are both spacetime dependent. MGR has the structure to describe anisotropies in galactic clusters in time and space that are in tension with the $ \Lambda$CDM model.
	\item MGR provides additional support for an inflationary model of the Universe. In the FLRW metric of MGR, the conservation equation followed by the Friedmann equations are
	\begin{equation}\label{CE}
		\dot{\varrho}-\frac{c^{2}}{8\pi G}\dot{\varPhi_{00}}=-\frac{3\dot{a}}{a}[\varrho+\frac{p}{c^{2}}-\frac{c^{2}}{8\pi G}(\varPhi_{00}+\frac{P_{d}}{c^{2}})],
	\end{equation}
	\begin{equation}\label{F1}
		\frac{\ddot{a}}{a}=-\frac{4\pi G}{3}(\varrho+\frac{3p}{c^{2}})+\frac{\varPhi_{00}c^{2}}{6}+\frac{P_{d}}{2}
	\end{equation}and
	\begin{equation}\label{F2}
		(\frac{\dot{a}}{a})^{2}=\frac{8\pi G\varrho}{3}-\kappa\frac{c^{2}}{a^{2}}-\frac{c^{2}\varPhi_{00}}{3}
	\end{equation}where the dot refers to the derivative with respect to $ x_{0} $ and $ \kappa $ is a parameter representing one of three spatial geometries in the FLRW metric. $ a(t) $ is the cosmological scale factor, which satisfies $ a>0 $ after the Big Bang at $ t=0 $. $ P_{d} $ is proportional to the gravitational pressure. With the Hubble parameter defined as $ H:=\frac{\dot{a}}{a} $, it follows that $ \frac{\ddot{a}}{a}=\dot{H}+H^{2} $ and equations (\ref{CE}), (\ref{F1}) and (\ref{F2}) combine to yield
	\begin{equation}\label{BB}
		-\frac{8\pi G}{c^{2}}\frac{d\varrho}{da}+\frac{d\varPhi_{00}}{da}=\frac{6\kappa}{a^{3}}-\frac{6\dot{H}}{ac^{2}}.
	\end{equation}If $ \dot{H} $ is constant, (\ref{BB}) has the solution
	\begin{equation}\label{p00}
		\varPhi_{00}=\Lambda+\frac{8\pi G\varrho}{c^{2}}-\frac{3\kappa}{a^{2}}-\frac{6\dot{H}\ln a}{c^{2}}
	\end{equation}where $ \Lambda $ is the cosmological constant, which appears early in the primordial epoch when there is no ordinary matter. Immediately following the Big Bang, an infinitesimal slice of a maximally symmetric Lorentzian spacetime with an enormous gravitational energy density, violently expands into a Universe with an extremely large scale factor.  
\end{enumerate}\par Thus in MGR, the cosmological constant appears as the constant term of the gravitational energy density in the FLRW metric. It is not part of the Einstein-Hilbert action functional for gravity; $ \Phi $ dynamically replaces it. The sign of the cosmological constant is determined from equation (\ref{F2}), which demands the inequality 
$\frac{8\pi G\varrho}{c^{2}}-\frac{3\kappa}{a^{2}}-\varPhi_{00}>0$ so that $ \Lambda<\frac{6\dot{H}\ln a}{c^{2}} $. A positive cosmological constant requires $0<\Lambda<\frac{6\dot{H}\ln a}{c^{2}} $ with $ \dot{H}>0 $; otherwise, $\dot{H}\leq0  $ demands $ \Lambda $ to be negative.
\par There is increasing evidence \cite{Pes} that the Hubble parameter has increased from $ \approx\text{67 km s}^{-1} \text{Mpc}^{-1}$ in the early Universe to $ \approx\text{74 km s}^{-1}\text{Mpc} ^{-1}$ in the present epoch. That means $ \dot{H}>0 $ and $ \Lambda $ is positive, which is in agreement with the observations \cite{RiessCC, Perl}.  \par  If $ \dot{H} $ is not constant after the Big Bang, $ \varPhi_{00} $ will have additional structure in accorance with (\ref{BB}), which cannot be determined until $ \dot{H} $ is known from astronomical data. \par Equation (\ref{BB}) can be written as
\begin{equation}\label{BD}
	-\frac{8\pi G}{c^{2}}\dot{\varrho}+\dot{\varPhi_{00}}=6H(\frac{\kappa }{a^{2}}-\frac{\dot{H}}{c^{2}}), 	
\end{equation} which describes the sum of the flows of the mass and gravitational energy. MGR provides a richer formalism than the $ \Lambda$CDM model to describe various aspects of cosmology.
\subsection{The dark matter skeleton of curved spacetime}
Each point of a Lorentzian spacetime contains a myriad of quantum-metric covectors from which $ \varPhi_{\alpha\beta} $ and $ \Phi$ are constructed. Spacetime is filled with dark matter, gravitational energy-momentum, and dark energy. Immediately after the Big Bang, intense radiation prevented ordinary matter from forming. Spacetime was maximally symmetric and the associated Killing vectors rendered $ \Phi=0 $. The intense radiation did not interact with dark matter. It gravitationally clumped and formed dense pockets or wells, which broke the symmetry of spacetime. Ordinary matter accumulated in the wells by gravity. As spacetime expanded due to the presence of dark energy, filaments of dark matter formed, and ordinary matter built-up at the junctions of the filaments where gravity was the strongest. Dark matter spun a web throughout spacetime, forming the skeleton where ordinary matter accumulated due to gravity. Dark energy filled the voids in the cosmic web, and pushed them apart, subject to the global constraint (\ref{intPhi}) that keeps the Universe in balance.
\par Although it is beyond the scope of this manuscript, what needs to be done is a detailed computer model of the growth of a Lorentzian spacetime with the existence of a myriad of line element covectors as described.
\subsection{Galactic anisotropies}
The Universe has been observed to be expanding faster in one region than another in galaxy clusters \cite{Mig}, which is in tension with the isotropy of the local Universe. Since dark energy can vary with time or location, the anisotropy in the Hubble parameter may be attributed to dark energy or a bulk flow, which can be investigated from (\ref{BD}). MGR can be used as a model to test the anisotropies in galactic clusters.
\par To first order, dark energy is described by the cosmological constant, which cannot vary with time or location. $\Phi$, which dynamically replaces the cosmological constant in MGR, solves that problem. The change in time of the dark energy density, to a statistical significance of $\sim 4\sigma$, was discussed in Refs. \cite{Ris,Zhao}. Although that is one standard deviation from a strong confirmation of the time dependence of dark energy, further observations will likely confirm that conjecture.\par Dark energy and dark matter are intertwined in the description of anisotropies of galactic clusters and other aspects of the cosmos. MGR contains that intrinsic dual structure; $ \varPhi_{\alpha\beta} $ leads to a description of dark matter, and its trace with respect to the metric describes dark energy locally from the condition $ \Phi>-2\varPhi_{00} $.
\subsection{Anisotropies in the CMB}
The Cosmic microwave background radiation is the remnant radiation from the hot early days of the Universe. Evidence for dark matter comes from cosmological measurements of anisotropies in the CMB. \par The conservation equation in the FLRW metric of GR and the Friedmann equations with a cosmological constant
\begin{equation}\label{LF}
	\begin{split}
		\dot{\varrho}=-\frac{3\dot{a}}{a}(\varrho+\frac{p}{c^{2}}),\;
		\frac{\ddot{a}}{a}=-\frac{4\pi G}{3}(\varrho+\frac{3p}{c^{2}})+\frac{\Lambda c^{2}}{3},
		(\frac{\dot{a}}{a})^{2}=\frac{8\pi G\varrho}{3}-\kappa\frac{c^{2}}{a^{2}}+\frac{c^{2}\Lambda}{3}
	\end{split}
\end{equation}
are fundamental to the $ \Lambda $CDM model used to study the anisotropies of the CMB. However, the corresponding equations of MGR offer a richer environment with which to explore those anisotropies. That is immediately apparent by setting the gravitational energy-density $ \varPhi_{00}$ (in units of $L^{-2}  $) to the constant value $ -\Lambda $ and the gravitational pressure-density $ P_{d} $ to $ c^{2}\Lambda $ in equations (\ref{CE}), (\ref{F1}) and (\ref{F2}), which generates equations (\ref{LF}), respectively. $ \Phi $ is the energy-momentum density of the gravitational field (in units $ L^{-2} $), which incorporates both $ \varPhi_{00}$ and $ P_{d} $: $ \Phi=-\varPhi_{00}+\frac{3P_{d}}{c^{2}} $. $ \Phi $ dynamically replaces the cosmological constant in MGR and has the constant value of $ 4\Lambda $ to locally mimic the noted equations in the FLRW metric of GR.\par The dark energy density parameter of the $ \Lambda $CDM model $ \Omega_{\Lambda} $ and the dark matter density parameter $ \Omega_{d} $ can both be replaced by a new parameter $ \Omega_{\Phi} $ in MGR; subject to the constraint (\ref{intPhi}):$ \int \Phi\sqrt{-g}d^{4}x=0 $. The constraint does not permit a Universe described by a Lorentzian metric to rip apart from uncontrolled expansion, or to contract to oblivion.\par The techniques of the $ \Lambda $CDM model can then be applied to the total density parameter $ \Omega $, which includes $ \Omega_{\Phi} $ but neither $ \Omega_{d} $ nor $ \Omega_{\Lambda} $. Detailed calculations of the perturbations from a flat FLRW metric of both $ \varPhi_{00} $ and $ P_{d} $ would be required to obtain the angular temperature variations $ \frac{\Delta T}{T} $ of the CMB in terms of the spherical harmonic functions. Although the dark matter density affects all peaks in the angular variations of  $ \frac{\Delta T}{T} $, the third peak is highly determined by dark matter; a third peak with a height greater than or equal to that of the second peak evidences that dark matter dominated the matter density in the hot dense plasma before recombination. Since $ \Phi $ incorporates both dark energy and dark matter, it is conjectured that MGR can produce the angular temperature variations of the $ \Lambda $CDM model. 
\section{Conclusion}
Einstein's original postulate of a total energy-momentum tensor, which includes both  gravitating matter and gravitational energy-momentum, is reinstated in MGR. From the ODT, a \emph{connection-independent} tensor $ \varPhi_{\alpha\beta} $ is introduced that represents the energy-momentum of the gravitational field. Using Lovelock's theorem and the ODT generates the \emph{complete} Einstein equation in one line, which exemplifies the mathematical beauty of MGR. $ \varPhi_{\alpha\beta} $ solves the problem of the non-localization of gravitational energy-momentum in GR, preserves the ontology of the Einstein equation and maintains the equivalence principle. \par The line element field vectors $\bm{X}  $ and $\bm{u}  $ provide extra freedom to geometrically describe dark energy and the missing invisible mass attributed to dark matter. The extended Schwarzschild solution is derived from the matter-free Einstein equation of MGR, which depends on the line element field covectors. Gravity gravitates, and the static gravitational energy density is shown to be twice the Newtonian result plus a contribution involving dark energy.\par The modified Newtonian force contains two additional terms; one depends on dark matter and the other on dark energy. When the dark energy force exactly balances the Newtonian force, the power-four Tully-Fisher relation follows. Using that, the invisible mass halo of galaxy NGC 3198 within its extended flat rotation curve is calculated to be identical to that obtained in GR with the NFW dark matter profile. \par  Since dark matter and dark energy exist at every point in a Lorentzian spacetime,  it is not possible for galaxies to be purely baryonic and void of dark matter. However, the Newtonian term in the modified Newtonian force can dominate the dark forces. That is evident for the ultra-diffuse galaxy AGC 114905, where the Newtonian force is about 17 times larger than the sum of the two dark forces. \par The three classic tests of GR are calculated in the extended Schwarzschild metric to compare the testable differences due to the line element covectors in MGR. In the solar system, the dark matter correction in MGR competes with the quadrupole moment of the Sun, but the dark matter correction to the Shapiro time delay is minute. However, for the gravitational lensing of galactic clusters, MGR provides significant corrections to the GR result without a dark matter profile. The strong lensing from galactic cluster SDSS J0900+2234 with an Einstein ring calculated in MGR compares well with that obtained in GR with the NFW dark matter profile, but does suggest the dark halo mass is overestimated with that profile. Moreover, the Shapiro time delay for the Crab pulsar calculated in MGR is larger than that determined in GR with typical dark matter profiles commensurate with the fact that those dark matter profiles are not valid in the interval [5,30]\,kpc from the galactic centre of the Milky Way. \par The observed anisotropy in the Hubble parameter may be attributed to dark energy, which is a natural part of MGR. It predicts the space and time dependence of the cosmological constant because the scalar $ \Phi $ dynamically replaces $ \Lambda $. The intrinsic dual structure of MGR provides a rich formalism for the description of anisotropies of galactic clusters and other aspects of the cosmos. \par 
The geometry of spacetime is fundamental to a metric theory of gravity; the effective force of gravity from all sources of matter and energy is entirely determined by the curvature of spacetime. However, there is no conflict between the geometrical and particle descriptions of dark matter because the unit line element covectors of the Lorentzian metric satisfy the Klein-Gordon wave equation that describes spin-1 particles. Those covectors are ignored in General Relativity, and it cannot explain dark matter without introducing a matter profile to represent it.
\section*{Acknowledgments}
I would like to thank the anonymous referees for their constructive comments to improve the quality of the manuscript and Professors Maurice Dupre and Olga Gil-Medrano for discussions on the ODT. 
\appendix
\numberwithin{equation}{section}
\section {Orthogonal Decomposition Theorem}
\begin{theorem}
	A non-divergenceless (0,2) symmetric tensor $ \gamma_{\alpha\beta} $ in the symmetric cotangent bundle $ \Gamma^{\infty}_{c} (S^{2}T^{\ast}\mathcal{M}) $ with smooth sections of compact support on an n-dimensional noncompact paracompact boundaryless time-oriented Lorentzian manifold $ \mathcal{M} $ with a Levi-Civita connection can be orthogonally decomposed as $
	\gamma_{\alpha\beta}= v_{\alpha\beta}+ \varPhi_{\alpha\beta} $ where $v_{\alpha\beta}  $ represents a linear sum of symmetric divergenceless (0,2) tensors and $\varPhi_{\alpha\beta}=\frac{1}{2}\mathcal{L}_{X}g_{\alpha\beta}+\mathcal{L}_{X}(u_{\alpha}u_{\beta})$. The timelike unit vector $\bm{u}  $ is collinear with one of the pair of regular vectors in the line element field $ (\bm{X},-\bm{X}) $ and $\bm X$ is not a Killing vector.
\end{theorem}
\begin{proof}
	Let the Lorentzian manifold with metric $ (\mathcal{M},g_{\alpha\beta}) $ be noncompact paracompact, time-oriented, and boundaryless. A smooth regular line element field $(\bm{X},\bm{-X)}$ exists such that $\bm{X}$ is not a Killing vector field, as does a timelike unit vector $ \bm{u} $ collinear with one of the pair of line element vectors. Let $ \mathcal{M} $ be endowed with a smooth Riemannian metric $ g^{+}_{\alpha\beta} $. The smooth Lorentzian metric $ g_{\alpha\beta} $ is constructed from $ g^{+}_{\alpha\beta} $ and the unit covectors $ u_{\alpha}$ and $ u_{\beta}$ : $g_{\alpha\beta}=g^{+}_{\alpha\beta}-2u_{\alpha}u_{\beta} $. Let $\gamma $ and $ v $ belong to $\Gamma^{\infty}_{c} (S^{2}T^{\ast}\mathcal{M}) $, the cotangent bundle of symmetric $(0,2)$ tensors on $ \mathcal{M} $ with smooth sections of compact support, and let $K$ be a compact subset of $\Gamma^{\infty}_{c} (S^{2}T^{\ast}\mathcal{M})$  containing the support of $\delta\gamma $ such that $\delta\gamma$ is orthogonal to every $1$-form $\omega$ with the property that $\delta\delta^*\omega$ restricted to $K$ vanishes. A non-divergenceless $ (0,2) $ symmetric tensor $ \gamma_{\alpha\beta} $ can be orthogonally and uniquely decomposed with respect to $ g^{+}_{\alpha\beta} $ \cite{Berger,Gil} according to $ \gamma_{\alpha\beta}=v_{\alpha\beta}+\frac{1}{2}{\mathcal L}_Xg^{+}_{\alpha\beta}$ where $v_{\alpha\beta}  $ represents a linear sum of symmetric divergenceless (0,2) tensors and $ {\nabla^{+}}^{\alpha}v_{\alpha\beta}=0 $.
	\par The condition for decomposable with respect to $g^{+}$, $\delta\delta^{\star}\omega=0$, generates the wave equation $\square^{+} X_{\beta}=-\nabla_{\beta}^{+}\nabla_{\mu}^{+}X^{\mu}-R^{+\lambda}_{\beta}$, for which solutions exist \cite{fland}. 
	\par The divergence of $ v_{\beta}^{\alpha} $ in the mixed tensor bundle can be written as $ \nabla_{\alpha}v_{\beta}^{\alpha}=\partial_{\alpha}v_{\beta}^{\alpha}+\frac{v_{\beta}^{\alpha}}{2g}\partial_{\alpha}g-\frac{1}{2}v^{\alpha\lambda}\partial_{\beta}g_{\alpha\lambda} $ where $g=det(g_{\alpha\beta})$. This generates the (0,1) tensor
	\begin{equation}\label{Dv}
		\begin{split}
			\nabla_{\alpha}v_{\beta}^{\alpha}-\nabla^{+}_{\alpha}v_{\beta}^{\alpha}=v_{\beta}^{\alpha}\partial_{\alpha}(\ln\sqrt{-g}-\ln\sqrt{g^{+}})+v^{\alpha\lambda}\partial_{\beta}(u_{\alpha}u_{\lambda})
		\end{split}
	\end{equation} using (\ref{gab}). The metric compatibility of $ g_{\alpha\lambda}$ demands $v^{\alpha\lambda}\nabla_{\beta}g_{\alpha\lambda}=0$. In Riemann normal coordinates, $v^{\alpha\lambda}\partial_{\beta}(u_{\alpha}u_{\lambda})=0 $ holds for all $v^{\alpha\lambda} $, and $g=-g^{+}=-1$. The right-hand side of (\ref{Dv}) vanishes, which implies  $\nabla_{\alpha}v_{\beta}^{\alpha}-\nabla^{+}_{\alpha}v_{\beta}^{\alpha}=0 $ in all coordinate systems. Thus, $\nabla_{\alpha}v^{\alpha}_{\beta}=0 $ since $\nabla^{+}_{\alpha}v_{\beta}^{\alpha}=0 $ and
	\begin{equation}\label{decomp}
		\gamma_{\alpha\beta}=v_{\alpha\beta}+\frac{1}{2}{\mathcal L}_Xg_{\alpha\beta}+{\mathcal L}_X(u_{\alpha}u_{\beta})
	\end{equation}with $\nabla^{\alpha}v_{\alpha\beta}=0 $. Using $ X^{\lambda}=fu^{\lambda} $ where $f\neq0  $ is the magnitude of $X^{\lambda}  $, the expression  $X^{\lambda}\nabla_{\lambda}(u_{\alpha}u_{\beta}) $ in the second term of (\ref{decomp}) then vanishes in an affine parameterization and $
	\gamma_{\alpha\beta}=v_{\alpha\beta}+\varPhi_{\alpha\beta}	$ where $			\varPhi_{\alpha\beta}:=\frac{1}{2}(\nabla_{\alpha}X_{\beta}+\nabla_{\beta}X_{\alpha})+u^{\lambda}(u_{\alpha}\nabla_{\beta}X_{\lambda}+u_{\beta}\nabla_{\alpha}X_{\lambda})$.
	Provided $ \gamma_{\alpha\beta}\neq0 $, the decomposition is orthogonal: $<v_{\alpha\beta},\varPhi_{\alpha\beta}>=0  $. Since $\bm X$ is not a Killing vector, $\varPhi_{\alpha\beta}$ does not vanish.
\end{proof}
\section {Variations of the action functional}
There are three variables in MGR: $g^{\alpha\beta}$, $ X^{\beta} $ and $ u^{\beta} $. However, since $\bm{X}$ and $\bm{u}$ are collinear, $ X^{\beta}=fu^{\beta} $ and $X_{\beta}=fu_{\beta}  $ where $ f\neq0 $ is the magnitude of  both $ \bm{X} $ and its covector. Strictly speaking, the independent variables are $g^{+\alpha\beta}$, $ u^{\beta} $ and $ f $. However, since $ \frac{\delta }{\delta g^{+\alpha\beta}}=\frac{\delta }{\delta g^{\mu\nu}}\frac{\delta g^{\mu\nu}}{\delta g^{+\alpha\beta}}=\frac{\delta }{\delta g^{\alpha\beta}} $, $ g^{\alpha\beta} $ can be treated as an independent variable in the variation with respect to the inverse metric. Variations of $ S $ with respect to $g^{\alpha\beta}$, $ u^{\beta} $ and $ f $ are developed as follows:
\subsection*{Variation of $S^{G}$ with respect to $g^{\alpha\beta}$}
Variation of $ S^{EH} $ is well known from any textbook on GR. What needs to be established is the variation of $ S^{G} $ with respect to the inverse metric where 
$ S^{G}=-a\int\Phi\sqrt{-g}d^{4} $. The parameter $ a $ is not needed in this calculation, but does belong in $ S $. $ \Phi $ can be expressed as $\Phi=\nabla_{\alpha}X_{\beta}(g^{\alpha\beta}+2u^{\alpha}u^{\beta})=\nabla_{\alpha}X_{\beta}g^{+\alpha\beta}  $ so
\begin{equation}\label{dgPhi}
	\begin{split}
		-\delta \int\Phi\sqrt{-g}d^{4}x
		=-\int\delta (\nabla_{\alpha}X_{\beta})g^{+\alpha\beta}\sqrt{-g}d^{4}x-\int\nabla_{\alpha}X_{\beta}(\delta g^{\alpha\beta}+2\delta (u^{\alpha}u^{\beta}))\sqrt{-g}d^{4}x\\
		-\int\nabla_{\alpha}X_{\beta}(g^{\alpha\beta}+2u^{\alpha}u^{\beta})\delta\sqrt{-g}d^{4}x
	\end{split}
\end{equation}
The third integral in (\ref{dgPhi}) is equivalent to
\begin{equation}\label{3}
	\begin{split}
		\frac{1}{2}\int\nabla_{\mu}X_{\nu}(g^{\mu\nu}+2u^{\mu}u^{\nu})g_{\alpha\beta}\delta g^{\alpha\beta}\sqrt{-g}d^{4}x.
	\end{split}
\end{equation} \par To compute the second integral in (\ref{dgPhi}), consider $g_{\alpha\lambda} u^{\alpha}u^{\beta}=u_{\lambda}u^{\beta} $. In a Lorentzian spacetime $ \bm{u} $ satisfies $ u_{\lambda}u^{\lambda}=-1 $ but in a Riemannian spacetime, $ u_{\lambda}u^{\lambda}=1 $. Then $ \delta(u_{\lambda}u^{\beta})=0 $ 
\footnote{In $ g^{+}_{\alpha\beta}$, $u_{\lambda}u^{\lambda}=1  $ so in $ g_{\alpha\beta} $: $u_{\lambda}u^{\beta}=g^{\beta\mu}u_{\mu}u_{\lambda}\\=(g^{+\beta\mu}-2u^{\beta}u^{\mu})u_{\lambda}u_{\mu}=u_{\lambda}u^{\beta}-2u_{\lambda}u^{\beta}=-u_{\lambda}u^{\beta}$. Thus, if $\lambda=\beta, -1=-1 $ and if $\lambda\neq\beta, u_{\lambda}u^{\beta}=0  $ and $ \delta(u_{\lambda}u^{\beta})=0 $.} and $ \delta(u^{\alpha}u^{\beta})g_{\alpha\lambda}=-u^{\alpha}u^{\beta}\delta g_{\alpha\lambda} $.
Contracting that with $ g^{\lambda\rho} $ and re-labelling the indicies gives
\begin{equation}\label{deltauu}
	\delta(u^{\alpha}u^{\beta})=u^{\beta}u_{\lambda}\delta g^{\alpha\lambda}.
\end{equation}The second integral can then be expressed as $ -\int\nabla_{\alpha}X_{\beta}(\delta g^{\alpha\beta}+2u^{\beta}u_{\lambda}\delta g^{\alpha\lambda})\sqrt{-g}d^{4}x $. Keeping $ \alpha $ and $ \beta $ fixed while summing over $ \lambda $ in the integrand: if $ \lambda=\beta$, $ \nabla_{\alpha}X_{\beta}(\delta g^{\alpha\beta}+2u^{\lambda}u_{\lambda}\delta g^{\alpha\beta})=-\nabla_{\alpha}X_{\beta}\delta g^{\alpha\beta}$, and when $ \lambda\neq\beta $ the integrand is $\nabla_{\alpha}X_{\beta}\delta g^{\alpha\beta}  $ so the integrand vanishes in the sum over all $ \lambda $.
\par The first integral in (\ref{dgPhi}) is now calculated.
\begin{equation}\label{1}
	\begin{split}
		-\int\delta(\nabla_{\alpha}X_{\beta})g^{+\alpha\beta}\sqrt{-g}d^{4}x=-\int(\nabla_{\alpha}\delta X_{\beta}-X_{\lambda}\delta\Gamma^{\lambda}_{\alpha\beta})g^{+\alpha\beta}\sqrt{-g}d^{4}x\\
		=2\int\nabla_{\alpha}(u^{\alpha}u^{\beta})\delta X_{\beta}\sqrt{-g}d^{4}x+\int X_{\lambda}\delta\Gamma^{\lambda}_{\alpha\beta}g^{+\alpha\beta}\sqrt{-g}d^{4}x
	\end{split}
\end{equation}integrating by parts. The second integral in (\ref{1}) is expressed as
\begin{equation}\label{I2}
	\frac{1}{2}\int X_{\lambda}g^{+\alpha\beta}g^{\lambda\rho}(\nabla_{\alpha}\delta g_{\rho\beta}+\nabla_{\beta}\delta g_{\rho\alpha}-\nabla_{\rho}\delta g_{\alpha\beta})\sqrt{-g}d^{4}x.
\end{equation}To compute this integral, we need to use:\newline $ \delta(g^{+\alpha\beta}g_{\alpha\lambda})=\delta(\delta^{\beta}_{\lambda}+2u_{\lambda}u^{\beta})=0 $ so $ \delta g^{+\alpha\beta}g_{\alpha\lambda}=-g^{+\alpha\beta}\delta g_{\alpha\lambda} $. Contracting that with $ g^{\lambda\rho} $ gives
\begin{equation}\label{g+1}
	\delta g^{+\rho\beta}=-g^{+\alpha\beta}g^{\lambda\rho}\delta g_{\alpha\lambda}.
\end{equation}
Continuing with (\ref{I2}), we can use the Lorentz connection $ \nabla $ with $ g^{+} $ in the first two terms because it can operate on the divergence of a symmetric tensor as shown in the ODT:
\begin{equation}\label{key}
	\begin{split}
		-\frac{1}{2}\int X_{\lambda}
		(\nabla_{\alpha}\delta g^{+\alpha\lambda}+\nabla_{\beta}\delta g^{+\lambda\beta})\sqrt{-g}d^{4}x-\frac{1}{2}\int X_{\lambda}(g^{\alpha\beta}+2u^{\alpha}u^{\beta})g^{\lambda\rho}\nabla_{\rho}\delta g_{\alpha\beta}\sqrt{-g}d^{4}x\\
		=-\int X_{\lambda}
		\nabla_{\alpha}\delta g^{+\alpha\lambda}\sqrt{-g}d^{4}x+\frac{1}{2}\int X_{\lambda}g_{\alpha\beta}\nabla^{\lambda}\delta g^{\alpha\beta}\sqrt{-g}d^{4}x-\int X_{\lambda}u^{\alpha}u^{\beta}\nabla^{\lambda}\delta g_{\alpha\beta}\sqrt{-g}d^{4}x\\
		=-\int X_{\beta}
		\nabla_{\alpha}\delta (g^{\alpha\beta}+2u^{\alpha}u^{\beta})\sqrt{-g}d^{4}x+\frac{1}{2}\int X_{\lambda}g_{\alpha\beta}\nabla^{\lambda}\delta g^{\alpha\beta}\sqrt{-g}d^{4}x\\+\int X_{\lambda}u_{\alpha}u_{\beta}\nabla^{\lambda}\delta g^{\alpha\beta}\sqrt{-g}d^{4}x
	\end{split}
\end{equation}using $ u^{\alpha}u^{\beta}\nabla^{\lambda}\delta g_{\alpha\beta}=u_{\mu}u_{\nu}\nabla^{\lambda}g^{\alpha\mu}g^{\beta\nu}\delta g_{\alpha\beta}=-u_{\alpha}u_{\beta}\nabla^{\lambda}\delta g^{\alpha\beta} $. Integrating by parts gives
\begin{equation}\label{key}
	\begin{split}
		\int [\nabla_{\alpha}X_{\beta}-\frac{1}{2}\nabla^{\lambda}X_{\lambda}g_{\alpha\beta}-\nabla^{\lambda}(X_{\lambda}u_{\alpha}u_{\beta})]\delta g^{\alpha\beta}\sqrt{-g}d^{4}x+2\int \nabla_{\alpha}X_{\beta}\delta (u^{\alpha}u^{\beta})\sqrt{-g}d^{4}x\\
		=\int [\nabla_{\alpha}X_{\beta}-\frac{1}{2}\nabla^{\lambda}X_{\lambda}g_{\alpha\beta}-\nabla^{\lambda}(X_{\lambda}u_{\alpha}u_{\beta})+2u^{\lambda}u_{\beta}\nabla_{\alpha}X_{\lambda}]\delta g^{\alpha\beta}\sqrt{-g}d^{4}x
	\end{split}
\end{equation}from (\ref{deltauu}) and re-labelling indicies. \par The complete variation is then
\begin{equation}\label{key}
	\begin{split}
		\int [\nabla_{\alpha}X_{\beta}-\frac{1}{2}\nabla^{\lambda}X_{\lambda}g_{\alpha\beta}-\nabla^{\lambda}(X_{\lambda}u_{\alpha}u_{\beta})+2u^{\lambda}u_{\beta}\nabla_{\alpha}X_{\lambda}\\+\frac{1}{2}\nabla_{\mu}X_{\nu}(g^{\mu\nu}+2u^{\mu}u^{\nu})g_{\alpha\beta}]\delta g^{\alpha\beta}\sqrt{-g}d^{4}x\\
		+2\int\nabla_{\alpha}(u^{\alpha}u^{\beta})\delta X_{\beta}\sqrt{-g}d^{4}x.
	\end{split}
\end{equation}Since the variations in $ \delta g^{\alpha\beta} $ are independent to those of $ \delta X_{\beta} $, the variation in S will yield
\begin{equation}\label{nuabapp}
	\nabla_{\alpha}(u^{\alpha}u^{\beta})=0
\end{equation} and the contribution from the variation of $ S^{G} $ with respect to $ g^{\alpha\beta} $ is
\begin{equation}\label{key}
	\begin{split}
		\int[\nabla_{\alpha}X_{\beta}+2u^{\lambda}u_{\beta}\nabla_{\alpha}X_{\lambda}-\nabla^{\lambda}(X_{\lambda}u_{\alpha}u_{\beta})+\nabla_{\mu}X_{\nu}u^{\mu}u^{\nu}g_{\alpha\beta}]\delta g^{\alpha\beta}\sqrt{-g}d^{4}x\\
		=\int[\nabla_{\alpha}X_{\beta}+2u^{\lambda}u_{\beta}\nabla_{\alpha}X_{\lambda}-X^{\lambda}\nabla_{\lambda}(u_{\alpha}u_{\beta})+\nabla_{\mu}X_{\nu}(u^{\mu}u^{\nu}g_{\alpha\beta}-u_{\alpha}u_{\beta}g^{\mu\nu})]\delta g^{\alpha\beta}\sqrt{-g}d^{4}x\\
		=\int[\frac{1}{2}(\nabla_{\alpha}X_{\beta}+\nabla_{\beta}X_{\alpha})+u^{\lambda}(u_{\alpha}\nabla_{\beta}X_{\lambda}+u_{\beta}\nabla_{\alpha}X_{\lambda})\\+\nabla_{\mu}X_{\nu}(u^{\mu}u^{\nu}g_{\alpha\beta}-u_{\alpha}u_{\beta}g^{\mu\nu})]\delta g^{\alpha\beta}\sqrt{-g}d^{4}x
	\end{split}
\end{equation}using the symmetry in $ \delta g^{\alpha\beta} $ and setting $-X^{\lambda}\nabla_{\lambda}(u_{\alpha}u_{\beta})=-fu^{\lambda}\nabla_{\lambda}(u_{\alpha}u_{\beta})=0  $ in an affine parameterization where $ f\neq0 $ is the magnitude of $ X^{\lambda} $. The last term in the variation vanishes, which follows by writing the tensor in brackets   $-u_{\alpha}u_{\beta}g^{\mu\nu}+u^{\mu}u^{\nu}g_{\alpha\beta}$ as its equivalent, $\frac{1}{2}(g^{{+}{\mu\nu}}g_{\alpha\beta}-g^{+}_{\alpha\beta}g^{\mu\nu}) $, and choosing an orthonormal basis $(e_{\alpha})$ at a point $p\in\mathcal{M} $ for $ g^{+} $ with $ e_{0}=u $. Then, $ u^{0}=u_{0}=1 $, $ u^{i}=u_{i}=0\;(i=1,2,3)$, $g^{+}_{\alpha\beta}=\delta_{\alpha\beta}$, $ g^{00}=-g^{{+}00}$ and $g_{00}=-g^{+}_{00}$, with all other components of the metric g equal to those of the metric $ g^{+}$. From the definition of $\varPhi_{\alpha\beta}  $
\begin{equation}\label{}
	\varPhi_{\alpha\beta}:=\frac{1}{2}(\nabla_{\alpha}X_{\beta}+\nabla_{\beta}X_{\alpha})+u^{\lambda}(u_{\alpha}\nabla_{\beta}X_{\lambda}+u_{\beta}\nabla_{\alpha}X_{\lambda})
\end{equation}it follows that 
$ \delta{S}^{G}=-a\delta \int\Phi\sqrt{-g}d^{4}x=a\int\varPhi_{\alpha\beta}\delta g^{\alpha\beta}\sqrt{-g}d^{4}x.$
\subsection*{Variation of $S$ with respect to $f$}
$ f\neq0 $, the magnitude of $ X^{\beta} $, is independent of $ u^{\beta} $ so variation of $ S $ with respect to $ f $ must be performed. With $ X^{\beta} $ and its first derivative not appearing in the matter energy-momentum tensor, variation with respect to $ f $ involves only $ S^{G} $:
\begin{equation}\label{key}
	\begin{split}
		S^{G}=-a\int[\nabla_{\alpha}(fu^{\alpha})+2u^{\alpha}u^{\beta}\nabla_{\alpha}(fu_{\beta})]\sqrt{-g}d^{4}x\\
		=-a\int(f\nabla_{\alpha}u^{\alpha}-u^{\alpha}\nabla_{\alpha}f)\sqrt{-g}d^{4}x
	\end{split}
\end{equation}
using $ X_{\alpha}=fu_{\alpha}$ and $u^{\beta}\nabla_{\alpha}u_{\beta}=0$. Hence, 
\begin{equation}\label{key}
	\begin{split}
		\delta S^{G}&=-a\int(\nabla_{\alpha}u^{\alpha}\delta f-u^{\alpha}\delta\nabla_{\alpha}f)\sqrt{-g}d^{4}x\\
		&=-2a\int\nabla_{\alpha}u^{\alpha}\delta f\sqrt{-g}d^{4}x.
	\end{split}
\end{equation} It follows that 
\begin{equation}\label{bnu0}
	\nabla_{\alpha}u^{\alpha}=0 
\end{equation} for arbitrary variations of f. $ \nabla_{\alpha}u^{\alpha}=0 $ can also be obtained from (\ref{nuab}) in an affine parameterization.
\subsection*{Variation of $S$ with respect to $u^{\nu}$} 
The action functionals $S^{F}$ and $S^{EH}$ do not depend on $X^{\mu}$ and  $\frac{\delta X^{\mu}}{\delta u^{\nu}}=f\delta^{\mu}_{\nu}$ is well defined, so the variation with respect to $u^{\nu}$ of $S$ depends only on $S^{G}$:
\begin{equation}\label{Su}
\frac{\delta S^{G}}{\delta u^{\nu}}=-\int\frac{\delta }{\delta u^{\nu}}(\Phi\sqrt{-g})d^{4}x.
\end{equation}
From the collinearity $X_{\alpha}=fu_{\alpha}$, $\varPhi_{\alpha\beta}$ can be expressed as
\begin{equation}\label{phiu}
	\varPhi_{\alpha\beta}=\frac{f}{2}(\nabla_{\alpha}u_{\beta}+\nabla_{\beta}u_{\alpha})-\frac{1}{2}(u_{\beta}\nabla_{\alpha}f+u_{\alpha}\nabla_{\beta}f).
\end{equation}Using (\ref{bnu0}), $\Phi$ can be expressed as
\begin{equation}\label{Phiu}
	\Phi=-u^{\alpha}\nabla_{\alpha}f.
\end{equation}
Then, (\ref{Su}) becomes
\begin{equation}\label{dSu}
		\frac{\delta S^{G}}{\delta u^{\nu}}=\int[(\delta^{\alpha}_{\nu}\partial_{\alpha}f+u^{\alpha}\frac{\delta}{\delta u^{\nu}}(\partial_{\alpha}f) )\sqrt{-g}-\Phi\frac{\delta\sqrt{-g}}{\delta u^{\nu}}]d^{4}x
	\end{equation} The third term is
\begin{equation}\label{sqrtg}
	\begin{split}
		-\Phi\frac{\delta\sqrt{-g}}{\delta u^{\nu}}&=-\Phi\frac{\delta\sqrt{-g}}{\delta g^{\alpha\beta}}\frac{\delta g^{\alpha\beta}}{\delta u^{\nu}}\\
		&=-\Phi g_{\alpha\beta}\sqrt{-g}(\delta^{\alpha}_{\nu}u^{\beta}+\delta^{\beta}_{\nu}u^{\alpha})\\
		&=-2\Phi\sqrt{-g}u_{\nu}
	\end{split}
\end{equation}so (\ref{dSu}) can be written as
\begin{equation}
	\frac{\delta S^{G}}{\delta u^{\nu}}=\int(\partial_{\nu}f+u^{\alpha}\frac{\delta}{\delta u^{\nu}}(\partial_{\alpha}f)-2\Phi u_{\nu})\sqrt{-g})d^{4}x
\end{equation}

 Setting $\partial_{\alpha}f=u_{\alpha}P$ for some scalar $P$ requires
\begin{equation}
	u^{\alpha}\frac{\delta}{\delta u^{\nu}}(\partial_{\alpha}f)=-\frac{\delta P}{\delta u^{\nu}}-P u_{\nu}
\end{equation}using 
\begin{equation}
	u^{\alpha}\frac{\delta u_{\alpha}}{\delta u^{\nu}}=-u_{\nu}
\end{equation} from $u_{\alpha}u^{\alpha}=-1$. Thus, 
\begin{equation}
\frac{\delta S^{G}}{\delta u^{\nu}}=-\int(\frac{\delta \Phi}{\delta u^{\nu}}\sqrt{-g}+\Phi \frac{\delta\sqrt{-g}}{\delta u^{\nu}})d^{4}x=\int(\partial_{\nu}f-\frac{\delta P}{\delta u^{\nu}}-(P+2\Phi) u_{\nu})\sqrt{-g})d^{4}x=0.
\end{equation}Choosing $P=\Phi$ generates
\begin{equation}\label{bunu}
	\partial_{\nu}f=\Phi u_{\nu}.
\end{equation}
\par Note that (\ref{bnu0}) and (\ref{bunu}) are consistent with $ \Phi=\nabla_{\alpha}X^{\alpha}+2u^{\alpha}u^{\beta}\nabla_{\alpha}X_{\beta} $:\begin{equation}\label{PhiPhi}
	\begin{split}
		\Phi&=\nabla_{\alpha}X^{\alpha}+2u^{\alpha}u^{\beta}\nabla_{\alpha}X_{\beta}\\
		&=f\nabla_{\alpha}u^{\alpha}+u^{\alpha}\partial_{\alpha}f+2u^{\alpha}u^{\beta}u_{\beta}\partial_{\alpha}f+2fu^{\alpha}u^{\beta}\nabla_{\alpha}u_{\beta}\\
		&=u^{\alpha}\partial_{\alpha}f-2u^{\alpha}\partial_{\alpha}f\\
		&=-u^{\alpha}u_{\alpha}\Phi\\
		&=\Phi.
	\end{split}
\end{equation}
\section{Derivation of the spin-1 Klein-Gordon equation from the line element field}
The action functional considered is: $ S^{1}=-\frac{1}{2}\int[\nabla_{\alpha}X_{\beta}\nabla^{\alpha}X^{\beta}+k^{2}X_{\beta}X^{\beta}]\sqrt{-g}d^{4}x $ where $ X^{\beta} $ is an arbitrary line element vector. It is worthwhile to express $ S^{1} $ in terms of the independent variables $ f $ and $ u_{\beta} $ as:
\begin{equation}\label{S^{1}}
	\begin{split}
		S^{1}=-\frac{1}{2}\int(f^{2}\nabla_{\alpha}u_{\beta}\nabla^{\alpha}u^{\beta}-\nabla_{\alpha}f\nabla^{\alpha}f
		-f^{2}k^{2})\sqrt{-g}d^{4}x\\
		=\frac{1}{2}\int (f^{2}u_{\beta}\square u^{\beta}+\nabla_{\alpha}f\nabla^{\alpha}f
		+f^{2}k^{2})\sqrt{-g}d^{4}x
	\end{split}
\end{equation}using $ u^{\beta}\nabla_{\alpha}u_{\beta}=0 $ and integrating by parts. The variation of $ S^{1} $ with respect to $ f $ is: 
\begin{equation}\label{dS1}
	\begin{split}
		\delta S^{1}=\int [fu_{\beta}(\square u^{\beta}-k^{2}u^{\beta})\delta f+\nabla^{\alpha}f\delta\nabla_{\alpha}f]\sqrt{-g}d^{4}x\\
		=\int [fu_{\beta}(\square u^{\beta}-k^{2}u^{\beta})-\square f)]\delta f\sqrt{-g}d^{4}x
	\end{split}
\end{equation}after integrating by parts. Since $ fu_{\beta}\neq0 $, it follows from the variation that
\begin{equation}\label{KGu}
	\square u^{\beta}=k^{2}u^{\beta}
\end{equation} provided $\square f=0$. \par From (\ref{unu}), $ \nabla_{\alpha}f=\Phi u_{\alpha} $, so $\square(fu^{\beta})=u^{\beta}\square f+f\square u^{\beta}$ in an affine parameterization. Imposing the constraint
\begin{equation}\label{con}
	\square f=0
\end{equation}allows the variation from (\ref{dS1}) to be expressed as
\begin{equation}\label{key}
	\int u_{\beta}[\square(fu^{\beta})-fk^{2}u^{\beta}]\sqrt{-g}d^{4}x=0
\end{equation}from which
\begin{equation}\label{KGX}
	\square X^{\beta}=k^{2}X^{\beta}
\end{equation}since $ u_{\beta}\neq0 $. $X^{\beta} $ satisfies the spin-1 KG equation provided its magnitude fulfils the homogeneous wave equation $ \square f=0 $ and (\ref{unu}) holds, which from (\ref{wf}) demands $ \pounds_{u}\Phi=0 $. Given $ \square u^{\beta}=k^{2}u^{\beta}$, $\square X^{\beta}=k^{2}X^{\beta}$ iff $\square f=0$ and $ \nabla_{\alpha}f=\Phi u_{\alpha} $  where $X^{\beta}=fu^{\beta}$.


\begin{thebibliography}{00}
	\bibitem{1} A.Einstein and M. Grossmann, Entwurf einer verallgemeinerten Relativit{\"a}tstheorie und einer Theorie der Gravitation I. Physikalischer Teil von Albert Einstein. II. Mathematischer Teil von Marcel Grossman, Leipzig and Berlin: B. G. Teubner (1913). 
	\bibitem{2} A. Einstein, Die Grundlage der allgemeinen Relativit{\"a}tstheorie, {\it Annalen der Physik} \textbf{354} 7 769-822 (1916). 
	\bibitem{3} A. Einstein, Vier Vorlesungen uber Relativitatstheorie gehalten im Mai 1921 an der Universit at Princeton, Vieweg F (1922) 85.
	\bibitem{4} E. Poisson and C. Will, {\it Gravity: Newtonian, Post-Newtonian, Relativistic}, (Cambridge University Press 2014), p293.
	\bibitem{5} R. Vishwakarma, Einstein and Beyond: A Critical Perspective on General Relativity, {\it Universe} (2014).  
	\bibitem{6} M.Maggiore, {\it Gravitational Waves Volume 1 Theory and Experiments}, (Oxford University Press, Oxford 2008), pp.30,33,58.  
	\bibitem{7} C. Misner, K. Thorne and J. Wheeler, {\it Gravitation} (Freeman, San Francisco 1973), p. 467. 
	\bibitem{8} A. Einstein, {\it Physics and reality} Journal of the Franklin Institute {\bf221}, 349-382 (1936). 
	\bibitem{9} G. Nash, Modified General Relativity, {\it Gen. Relativ. Gravit.} \textbf{51}, 53 (2019), arXiv:1904.10803v7 [gr-qc].
	\bibitem{10} C.M. Will, {\it Theory and Experiment in Gravitational Physics} (Cambridge University Press, New York 1993), p.127. 
	\bibitem{11} T. Jacobson and D. Mattingly, Phys Rev D \textbf{70}, 024003 (2004).
	\bibitem{12} C.T.J. Dodson  {\it Categories, Bundles and Spacetime Topology} (Springer Science+Business Media, Dordrecht, 1988), pp.167-168. 
	\bibitem{13} B. O'Neill {\it Semi-Riemannian Geometry with applications to relativity} (Academic Press Limited, London, 1983), p.149. 
	\bibitem{14} S. Hawking and G. Ellis {\it The large scale structure of space-time} (Cambridge University Press, Cambridge, 1973), p.39. 
	\bibitem{15} Y. Choquet-Bruhat {\it General Relativity and Einstein's Equations} (Oxford University Press, Oxford, 2009), pp. 373,389. 
	\bibitem{Love} D. Lovelock, The Einstein Tensor and Its Generalizations, {\it J. Math. Phys.} \textbf{12}, 498-501 (1971).
	\bibitem{Chen} Chiang-Mei, Chen, Jian-Liang, Liu and Nester J.M., Gravitational energy is well defined, {\it Int J Mod Phys D} \textbf{27}, No. 14 (2018).
	\bibitem{MM} W.H. McCrea and E.A. Milne,   Newtonian universes and the curvature of space,  {\it Quarterly Journal of Mathematics} \textbf{5}, 73?80 (1934).
	\bibitem{LL} L. Landau and E. Lifshitz  {\it The Classical Theory of Fields (third edition)} (Pergamon Press, Toronto, 1971), p.330.
	\bibitem{Riess} A. Riess et al., Type Ia Supernova Discoveries at $z>1$ From the Hubble Space Telescope: Evidence for Past Deceleration and Constraints on Dark Energy Evolution {\it ApJ} \textbf{607}, 665 (2004). 
	\bibitem{Milg} M. Milgrom, A modification of the Newtonian dynamics as a possible alternative to the hidden mass hypothesis,  {\it Astrophysical Journal} Part 1 \textbf{270}, 365-370 ( 1983). 
	\bibitem{MONDA0} T.H. Randriamampandry and C. Carignan, Galaxy mass models: MOND versus dark matter haloes, {\it MNRAS} \textbf{439}, 2132-2145 (2014).
	\bibitem{Sal} P. Salucci, Dark Matter in galaxies: leads to its Nature, arXiv:1302.2268v2 [astro-ph.GA] (2013). 
	\bibitem{vanD} P. van Dokkum et al. A galaxy lacking dark matter, {\it Nature} \textbf{555}, 629-632 (2018). 
	\bibitem{Pina} P. Pi\~{n}a, et al., No need for dark matter: resolved kinematics of the ultra-diffuse galaxy AGC 114905, {\it MNRAS} \textbf{512}, 3230?3242 (2022).
	\bibitem{Sell} J.A Sellwood and R.H. Sanders, The ultra-diffuse galaxy AGC 114905 needs dark matter, {\it MNRAS} \textbf{514}, 4008-4017 (2022).
	\bibitem{Kar} E. Karukes et al. The dark matter distribution in the spiral NGC 3198 out to 0.22 $R_{vir}  $, {\it A\&A} \textbf{578}, A13 ( 2015). 
	\bibitem{Clowe} D. Clowe et al. A Direct Empirical Proof of the Existence of Dark Matter, {\it ApJ} \textbf{648}, L109 (2006).
	\bibitem{Para} D. Paraficz et al., The Bullet cluster at its best: weighing stars, gas and dark matter, {\it Astron. Astrophys.} \textbf{594}, A121 (2018).
	\bibitem{Bar} B. Barman et al., Non-minimally coupled vector boson dark matter, {\it JCAP} \textbf{2022}, 01 (2022).
	\bibitem{NQ} G. Nash, Modified General Relativity and Quantum Theory in curved spacetime, {\it Int J Mod Phys A} \textbf{36}, No.29 (2021).
	\bibitem{burke} D.L Burke et al., Positron Production in Multiphoton Light-by-Light Scattering, {Phys. Rev. Lett.} \textbf{79}, 1626 (1997).
	\bibitem{Ciu} I. Ciufolini and J. Wheeler, {\it Gravitation and Inertia} (Princeton University Press, Princeton, 1995), p. 141.
	\bibitem{Solar} S. Pireaux, JP. Rozelot and S. Godier, Solar quadrupole moment and purely relativistic gravitation contributions to Mercury's perihelion Advance, arXiv:astro-ph/0109032v1 (2001).
	\bibitem{Abu}  R. Abuter et al., Detection of the Schwarzschild precession in the orbit of the star S2 near the Galactic centre massive black hole, {\it Astronomy and Astrophysics} \textbf{636}, L5 (2020).
	\bibitem{Wie} M.P. Wiesner,  et al., The Sloan Bright Arcs Survey: Ten strong gravitational lensing clusters and evidence of overconcentration, {\it Astrophys.J.} \textbf{761}, 1 (2012).
	\bibitem{Posti} L. Posti and A. Helmi, Mass and shape of the Milky Way?s dark matter halo with globular clusters from Gaia and Hubble, {\it A\&A} \textbf{621}, A56 (2019).
	\bibitem{Des} S. Desai and E. Kahya, Galactic Shapiro delay to the Crab pulsar and limit on weak equivalence principle violation, {\it Eur. Phys. J. C}, \textbf{78}, 86 (2018).
	\bibitem{Cautun} M. Cautun et al.,  The MilkyWay total mass profile as inferred from Gaia DR2, {\it MRAS}, \textbf{494} 2 (2020).
	\bibitem{Pes} D. W. Pesce et al., The Megamaser Cosmology Project. XIII. Combined Hubble constant constraints, {\it Astrophys. J. Lett.} \textbf{891}, Number 1 (2020).
	\bibitem{RiessCC} A. Riess et al., Observational Evidence from Supernovae for an 
	Accelerating Universe and a Cosmological Constant, {\it Astron. J.} \textbf{116}, 1009 (1998).
	\bibitem{Perl} S. Perlmutter et al., Measurements of $ \Omega $ and $ \Lambda $ from 42 high-redshift supernovae, {\it ApJ} \textbf{517}, 565 (1999).
	\bibitem{Mig} K. Migkas et al., Cosmological implications of the anisotropy of ten galaxy cluster scaling relations, {\it Astronomy \& Astrophysics} \textbf{649}, A151 ( 2021). 
	\bibitem{Ris} G. Risaliti and E. Lusso, Cosmological constraints from the Hubble diagram of quasars at high redshifts, {\it Nature Astronomy} \textbf{3}, 272-277 (2019). 
	\bibitem{Zhao} G.B. Zhao et al. Dynamical dark energy in light of the latest observations,  {\it Nature Astronomy} \textbf{1}, 627-632 (2017).
	\bibitem{Berger} M. Berger and D. Ebin, Some Decompositions of the Space of Symmetric Tensors on a Riemannian Manifold, {\it J. Differential  Geometry}, \textbf{3}, 379-392 (1969).
	\bibitem{Gil} O. Gil-Medrano and A. Montesinos Amilibia, About a Decomposition of the Space of Symmetric Tensors of Compact Support on a Riemannian Manifold {\it New York J. Math.} \textbf{1} 10-25 (1994).
	\bibitem{fland} F.G. Friedlander, The wave equation on a curved space-time Cambridge University Press, Cambridge, UK, 1975, p. 219.
\end{thebibliography}
\end{document}